\documentclass[11pt, a4paper]{article}
\usepackage[utf8]{inputenc}
\usepackage[T1]{fontenc}
\usepackage{amsmath, amssymb, amsthm}
\usepackage{amsfonts}
\usepackage{booktabs}
\usepackage{cite}
\usepackage{graphicx}
\usepackage{hyperref}
\usepackage{geometry}
\usepackage{xcolor}
\usepackage{enumitem}
\usepackage{multirow}
\usepackage{url}
\usepackage{tabularx}
\usepackage{tikz}
\usepackage{pgfplots}
\usepackage{algorithm}
\usepackage{algorithmic}
\usepackage{bm}
\usepackage{comment}
\usepackage{array}
\usepackage{authblk}              
\definecolor{PassGreen} {HTML}{27AE60}
\newtheorem{theorem}{Theorem}
\newtheorem{definition}{Definition}
\newtheorem{lemma}{Lemma}
\newtheorem{claim}{Claim}
\newtheorem{remark}{Remark}

\begin{document}

\title{When the Learning With Errors Problem Meets the Coherent Ising Machine: A Penalty-Free Algorithm-Hardware Co-Design}
\author{Shuxian Jiang\thanks{Email: jiangshuxian.grace@gmail.com}}
\date{}
\maketitle

\begin{abstract}
The Learning With Errors (LWE) problem constitutes the mathematical foundation of modern Post-Quantum Cryptography (PQC). Cryptanalysis of LWE ranges from classical lattice reduction to machine learning and quantum--classical hybrids.

We propose CIM-BDD, a hybrid Bounded-Distance-Decoding solver that reduces LWE to a Quadratic Unconstrained Binary Optimization (QUBO) problem through a strictly \emph{penalty-free} mapping. An algebraic elimination of the secret embeds LWE into a $q$-ary lattice, absorbing the modular arithmetic and recasting the problem as a Closest Vector Problem (CVP). The squared error norm is then used \emph{directly} as the QUBO energy, so the cryptographic noise is the objective to be minimized rather than a penalized constraint. To realize this general model on current Noisy Intermediate-Scale Quantum (NISQ) devices, we design a special encoding method: a Continuous Relaxed Babai's Nearest Plane (CR-BNP) projection drives an adaptive mixed-radix encoder that greatly reduces both the qubit count and the QUBO coefficient range, so that a single batched hardware submission suffices. We further derive a statistically bounded early-stopping threshold ($T_{\text{early}}$) that acts as a one-sided certificate and doubles as a Decision-LWE distinguisher.

We validate the framework on the TU Darmstadt LWE Challenge, giving an end-to-end demonstration for both Search- and Decision-LWE of a $40$-dimensional instance on the Coherent Ising Machine CPQC-550. This work establishes a new algorithm-hardware co-design paradigm for quantum-classical hybrid cryptanalysis.
\end{abstract}

\section{Introduction}
Unlike integer factorization and the discrete logarithm problem, the Learning With Errors (LWE) problem, introduced by Oded Regev~\cite{Regev05}, is believed to resist quantum attack. It is extensively utilized in Post-Quantum Cryptography (PQC) and provides the theoretical foundation for NIST standard schemes such as ML-KEM (FIPS 203) and ML-DSA (FIPS 204)~\cite{FIPS203, FIPS204}, and it is the fundamental building block of Fully Homomorphic Encryption (FHE). LWE enjoys a worst-case to average-case guarantee: solving the average-case problem is at least as hard as worst-case lattice problems such as the Shortest Vector Problem (SVP). Concretely, LWE asks to recover a secret vector $\mathbf{s} \in \mathbb{Z}_q^n$ from a matrix $\mathbf{A}$ and a vector $\mathbf{c}$ with
\begin{equation} \label{eq:lwe_def}
    \mathbf{c} = \mathbf{A}\mathbf{s} + \mathbf{e} \pmod{q},
\end{equation}
where $\mathbf{e}$ is a small error vector drawn from a discrete Gaussian distribution. This unknown error is what makes the system intractable for both classical and quantum algorithms.

Evaluating the concrete security of LWE is an active area of research. In classical cryptanalysis, the lattice-based method that maps LWE to Bounded Distance Decoding (BDD) and solves it via lattice basis reduction, notably BKZ~\cite{Schnorr1994, Chen2011, Aono2016}, remains dominant, complemented by advanced estimators based on enumeration with pruning or sieving~\cite{Albrecht2015, BDGL16}. Locating closest or shortest vectors is, however, (sub)exponential. In parallel, machine-learning attacks~\cite{Wenger2022, Li2023Picante} have emerged, but they rely on large training datasets and scale poorly to cryptographic dimensions~\cite{Stevens2024, Wenger2025}.

With the advent of NISQ devices, attention has shifted to quantum--classical hybrids. For instance, Lv \textit{et al.}~\cite{Lv2022} use QAOA to optimize the nearest plane algorithm in order to solve LWE and apply a Variational Quantum Eigensolver to the unique SVP (uSVP), while Zheng \textit{et al.}~\cite{Zheng2025} encode lattice vectors into an Ising model (HAWI). Reducing LWE to a lattice problem via nullspaces absorbs the modular arithmetic classically (as in HAWI~\cite{Zheng2025}), but framing the resulting problem as an SVP inherently introduces the trivial-zero symmetry. To avoid collapsing to the origin without explicit penalties, SVP-based solvers are forced to restrict variable domains and sequentially evaluate $\mathcal{O}(m)$ or $\mathcal{O}(n)$ independent sub-Hamiltonians~\cite{Zheng2025, DableHeath2023}, or attempt to target the first excited state~\cite{Joseph2020}. Conversely, direct LWE QUBO formulations~\cite{Qayyum2023, Qayyum2025} rely heavily on auxiliary slack variables and penalty coefficients to enforce modular constraints, which distort the energy landscape. Our method is similar to \cite{Lv2022}, but makes a different algorithm-hardware co-design. By modeling LWE as CVP (BDD) and using the squared error norm directly as the objective energy, we get a penalty-free QUBO formulation (as well as a justification), then partition the search space to model partial dimensions using a Continuous Relaxed Babai's Nearest Plane driven adaptive mixed-radix encoder. We present both theoretical analysis and experimental verification on real hardware for LWE challenge instance of dimension 40. 

\subsection{Our Contributions}
We propose the CIM-BDD framework, a hybrid BDD solver suitable for Coherent Ising Machines (CIMs)~\cite{McMahon2016, Wei2026}, which support fully-connected topologies. Our framework advances the field through both mathematical modeling and physical hardware realization:

\begin{itemize}
    \item \textbf{A Strictly Penalty-Free QUBO Formulation of LWE.}
    In Section~\ref{sec:modeling}, an algebraic elimination of the secret absorbs the modulus into a primal $q$-ary lattice and recasts LWE as a Closest Vector Problem (CVP), then uses the Babai-recentered least-squares residual as the QUBO energy. So the cryptographic noise becomes the objective rather than a penalized constraint. This yields a strictly penalty-free modeling.
    \item \textbf{A Compact Encoding for NISQ-Scale Qubit Budgets and Precision Limits.}
    To map this general model onto precision-limited hardware (Section~\ref{sec:hardware}), a Continuous Relaxed Babai's Nearest Plane (CR-BNP) projection acts as a prior driving an adaptive mixed-radix encoder (allocating $0/1/2$ bits). This compresses the target combinatorial space, restricting the $n=40$ instance to fewer than $20$ logical qubits.
    \item \textbf{A One-Sided Statistical Certificate and DLWE Distinguisher via $T_{\text{early}}$.}
    From the $\chi^2$ properties of the discrete-Gaussian error variance we derive an early-stopping threshold $T_{\text{early}} = m'\sigma^2 + 4\sigma^2\sqrt{2m'}$. Placing a four-sigma confidence boundary on the residual energy makes this a one-sided statistical certificate: a candidate that crosses below $T_{\text{early}}$ certifies the LWE secret directly from the ground-state energy, with no algebraic re-verification. The same binary crossing doubles as a Decision-LWE distinguisher.
\end{itemize}

\section{Preliminaries}

\subsection{Lattices and the LWE Problem}
A lattice is defined as the set of all integer linear combinations of linearly independent basis vectors. Let $\mathbf{b}_1, \mathbf{b}_2, \dots, \mathbf{b}_n \in \mathbb{R}^m$ (with $n \leq m$) be a set of linearly independent vectors. The lattice generated by these vectors is defined as:
\begin{equation}
    \mathcal{L}(\mathbf{b}_1, \dots, \mathbf{b}_n) = \left\{ \sum_{i=1}^n x_i \mathbf{b}_i : x_i \in \mathbb{Z} \right\}.
\end{equation}
The set of vectors $\{\mathbf{b}_1, \dots, \mathbf{b}_n\}$ is called a \textit{basis} of the lattice. It is often convenient to represent the basis as a matrix $\mathbf{B} = [\mathbf{b}_1, \mathbf{b}_2, \dots, \mathbf{b}_n] \in \mathbb{R}^{m \times n}$ having the basis vectors as its columns. Consequently, the lattice can be equivalently denoted as $\mathcal{L}(\mathbf{B}) = \{ \mathbf{B}\mathbf{x} : \mathbf{x} \in \mathbb{Z}^n \}$.

The security of lattice-based cryptography fundamentally relies on the presumed intractability of several well-known computational problems, like the Shortest Vector Problem (SVP) and the Closest Vector Problem (CVP). 

\begin{description}
    \item[Shortest Vector Problem (SVP):] Given a lattice basis $\mathbf{B}$, find a non-zero lattice vector $\mathbf{v} \in \mathcal{L}(\mathbf{B})$ such that its Euclidean norm $||\mathbf{v}||$ is minimized. The length of this shortest non-zero vector is known as the first successive minimum of the lattice, denoted as $\lambda_1(\mathcal{L})$.
    
    \item[Closest Vector Problem (CVP):] Given a lattice basis $\mathbf{B}$ and a target vector $\mathbf{t} \in \mathbb{R}^m$ (not necessarily in the lattice), find a lattice vector $\mathbf{v} \in \mathcal{L}(\mathbf{B})$ that minimizes the distance $||\mathbf{v} - \mathbf{t}||$.
\end{description}

SVP and CVP are the most foundational worst-case problems in lattice theory. Both problems are known to be NP-hard and remain computationally intractable for classical and quantum algorithms even when approximated within polynomial factors. However, cryptographic constructions generally require problems that are hard on \textit{average} rather than in the worst case. The Learning With Errors (LWE) problem bridges this gap.
\begin{definition}[Learning With Errors (LWE)]
Let $n, m, q$ be positive integers. Let $\mathbf{A} \in \mathbb{Z}_q^{m \times n}$ be a uniformly random matrix and $\mathbf{s} \in \mathbb{Z}_q^n$ be a secret vector. The LWE samples are given by the pair $(\mathbf{A}, \mathbf{c})$, where $\mathbf{c} = \mathbf{A}\mathbf{s} + \mathbf{e} \pmod{q}$, and $\mathbf{e} \in \mathbb{Z}^m$ is sampled independently from a discrete Gaussian distribution $\mathcal{D}_{\mathbb{Z}^m,\sigma}$ with standard deviation $\sigma = \alpha q$.
\end{definition}
The search variant (Search-LWE, SLWE) aims to find $\mathbf{s}$, the decision variant (Decision-LWE, DLWE) aims to distinguish $(\mathbf{A}, \mathbf{c})$ from uniform randomness.

\paragraph{$q$-ary Lattices}
Given an integer modulus $q \ge 2$ and a uniformly random public matrix $\mathbf{A} \in \mathbb{Z}_q^{m \times n}$ (with $m > n$), a primal $q$-ary lattice $\Lambda_q(\mathbf{A})$ is defined as the set of all integer vectors that belong to the image of $\mathbf{A}$ modulo $q$:
\begin{equation}
    \Lambda_q(\mathbf{A}) = \left\{ \mathbf{v} \in \mathbb{Z}^m \mathrel{\Big|} \exists \mathbf{s} \in \mathbb{Z}_q^n \text{ such that } \mathbf{v} \equiv \mathbf{A}\mathbf{s} \pmod q \right\}.
\end{equation}

\paragraph{LWE as Average-Case Bounded-Distance Decoding (BDD).}
The Bounded-Distance Decoding (BDD) problem is a variant of the closest vector problem (CVP), where the target point is guaranteed to be so close to the lattice that there is a unique closest vector. LWE can be seen as an average-case BDD problem over the $q$-ary lattice $\Lambda_q(\mathbf{A})$~\cite{LindnerPeikert2011}. Because any valid algebraic evaluation $\mathbf{A}\mathbf{s} \pmod q$ mathematically constitutes a true lattice point $\mathbf{v} \in \Lambda_q(\mathbf{A})$, the LWE target vector $\mathbf{c}$ represents a spatial coordinate shifted away from the lattice point $\mathbf{v}$ by the discrete Gaussian noise $\mathbf{e}$. Therefore, finding this closest lattice vector $\mathbf{v}$ isolates the error $\mathbf{e}$ and recovers the cryptographic secret $\mathbf{s}$.

\subsection{Quadratic Unconstrained Binary Optimization (QUBO)}
In combinatorial optimization, QUBO represents binary variables $\mathbf{x} \in \{0, 1\}^N$ minimizing $f(\mathbf{x}) = \mathbf{x}^\top \mathbf{Q} \mathbf{x}$, where $\mathbf{Q}$ is an upper-triangular real-weighted matrix. QUBO is mathematically isomorphic to the Ising Hamiltonian $H(\mathbf{z}) = - \sum_{i<j} J_{ij} z_i z_j - \sum_i h_i z_i$ for $z_i \in \{+1, -1\}$. Coherent Ising Machines (CIMs) support fully-connected graph topologies for these models via measurement-feedback architectures~\cite{McMahon2016, Hamerly2019}.

\section{The Exact Model: A Penalty-Free QUBO Formulation via a Least-Squares CVP Objective}
\label{sec:modeling}

This section establishes the mathematical core of the framework: a reduction
that turns LWE into an exact, penalty-free QUBO whose energy \emph{is} the
squared error norm. The guiding principle is to make the cryptographic noise
the quantity to be \emph{minimized} rather than a constraint to be
\emph{penalized}. Standard QA/CIM mappings of LWE either introduce auxiliary
variables to encode the modular arithmetic or add penalty/slack terms to
enforce feasibility, both of which distort the energy landscape. We instead remove
the modulus by algebraic elimination and define the objective as a
least-squares residual whose minimizer is, by maximum likelihood, the true
error vector. The result is a penalty-free cost function in which the
target sits at the unique global minimum and the trivial-zero symmetry is
broken by the residual's own linear term.

\subsection{Reduction from LWE to CVP}
We begin with the standard LWE relation over $\mathbb{Z}_q$: $\mathbf{A}\mathbf{s} + \mathbf{e} \equiv \mathbf{c} \pmod q$. To eliminate the secret vector $\mathbf{s}$, we define a matrix $\mathbf{W}$ whose columns form a basis for the left nullspace of $\mathbf{A}$ modulo $q$, satisfying $\mathbf{W}^\top \mathbf{A} \equiv \mathbf{0} \pmod q$. 
Multiplying the LWE relation by $\mathbf{W}^\top$ from the left yields:
\begin{equation}
    \mathbf{W}^\top \mathbf{c} \equiv \mathbf{W}^\top(\mathbf{A}\mathbf{s} + \mathbf{e}) = \mathbf{W}^\top \mathbf{A} \mathbf{s} + \mathbf{W}^\top \mathbf{e} \equiv \mathbf{W}^\top \mathbf{e} \pmod q
\end{equation}
Rearranging the terms, we obtain:
\begin{equation}
   \mathbf{W}^\top (\mathbf{e} - \mathbf{c}) \equiv \mathbf{0} \pmod q
\end{equation}

Eq. (5) indicates that the vector $(\mathbf{e} - \mathbf{c})$ lies precisely within the lattice spanned by the columns of $\mathbf{A}$ modulo $q$. Let $\mathbf{M}$ be the generating matrix of this $q$-ary lattice space, typically constructed as $\mathbf{M} = [\mathbf{A} \mid q\mathbf{I}]$. By explicitly extracting a full-rank integer basis for this lattice (as detailed in Appendix~\ref{sec:appendix_algebraic_basis}) and performing lattice basis reduction (such as LLL or BKZ), we obtain a reduced basis matrix $\mathbf{B}_0$.
Thus there exists an integer vector $\mathbf{y}$ satisfying $\mathbf{c} - \mathbf{e} = \mathbf{B}_0 \mathbf{y}$. 

The primary objective is to recover the error vector $\mathbf{e}$, which is sampled from a distribution with a small variance. This is mathematically equivalent to finding an integer vector $\mathbf{y}$ that minimizes the squared Euclidean norm of $\mathbf{e}$:
\begin{equation}
    \|\mathbf{e}\|^2 = \|\mathbf{c} - \mathbf{B}_0 \mathbf{y}\|^2 = \|\mathbf{B}_0 \mathbf{y} - \mathbf{c}\|^2
\end{equation}
This formulation represents a standard Closest Vector Problem (CVP), where the goal is to find a lattice point $\mathbf{B}_0 \mathbf{y}$ that is closest to the target vector $\mathbf{c}$.

\subsection{Search Space Recentering via Babai's Algorithm}
Babai's Nearest Plane algorithm\cite{Babai1985} provides an approximate closest lattice point, and we denote its corresponding coefficient vector as $\mathbf{y}_{\text{center}}$.

We parameterize the target integer vector $\mathbf{y}$ as an offset $\Delta \mathbf{y}$ from this center point:
\begin{equation}
    \Delta \mathbf{y} = \mathbf{y} - \mathbf{y}_{\text{center}} \implies \mathbf{y} = \Delta \mathbf{y} + \mathbf{y}_{\text{center}}
\end{equation}
Substituting this offset parameterization into our objective norm function, we obtain:
\begin{align}
    \|\mathbf{e}\|^2 &= \|\mathbf{B}_0 (\Delta \mathbf{y} + \mathbf{y}_{\text{center}}) - \mathbf{c}\|^2 \nonumber \\
            &= \|\mathbf{B}_0 \Delta \mathbf{y} + \mathbf{B}_0 \mathbf{y}_{\text{center}} - \mathbf{c}\|^2
\end{align}

To recenter the search space around the origin, we define a shifted target vector $\mathbf{c}_{\text{shift}}$ representing the residual error of the approximate center point:
\begin{equation}
    \mathbf{c}_{\text{shift}} = \mathbf{c} - \mathbf{B}_0 \mathbf{y}_{\text{center}}
\end{equation}
Substituting $\mathbf{c}_{\text{shift}}$, the distance minimization objective is transformed into:
\begin{equation}
    \|\mathbf{e}\|^2 = \|\mathbf{B}_0 \Delta \mathbf{y} - \mathbf{c}_{\text{shift}}\|^2
\end{equation}
By reformulating the problem with $\mathbf{c}_{\text{shift}}$, the subsequent search algorithm only needs to explore small displacement vectors $\Delta \mathbf{y}$ around the origin. This recentering technique significantly bounds the search radius and improves the overall computational efficiency of recovering the error vector $\mathbf{e}$.

\subsection{QUBO Objective Formulation for the LWE Problem}
We use $\|\mathbf{e}\|^2$ as the objective function to include the discrete Gaussian distribution of $\mathbf{e}$. This approach mathematically corresponds to performing Maximum Likelihood Estimation (MLE) under the assumption that the noise follows a discrete Gaussian distribution. Within the BDD regime the planted error is the unique closest offset of c to the lattice, so this least-squares minimizer recovers the true error with high probability. 

\begin{remark}[Statistical Equivalence via Maximum Likelihood Estimation]
\label{rem:mle}
For a Learning With Errors (LWE) instance with independent Gaussian noise, minimizing the squared Euclidean norm of the error vector $\|\mathbf{e}\|^2$ is mathematically equivalent to computing the Maximum Likelihood Estimate (MLE) of these variables.
\end{remark}

\begin{proof}
In the standard LWE setting, each component $e_i$ ($i = 1, \dots, m$) of the error vector $\mathbf{e}$ is sampled independently from a discrete Gaussian distribution over $\mathbb{Z}$ centered at $0$ with standard deviation $\sigma$. The probability function for a single error component $e_i$ is proportional to:
\begin{equation}
    P(e_i) \propto \exp\left( -\frac{e_i^2}{2\sigma^2} \right)
\end{equation}

Because the components $e_i$ are independent and identically distributed (i.i.d.), the joint probability distribution (or the likelihood function $\mathcal{L}$) of the entire error vector $\mathbf{e}$ is the product of the probabilities of its individual components:
\begin{equation}
    \mathcal{L}(\mathbf{e}) = \prod_{i=1}^m P(e_i) \propto \prod_{i=1}^m \exp\left( -\frac{e_i^2}{2\sigma^2} \right) = \exp\left( - \sum_{i=1}^m \frac{e_i^2}{2\sigma^2} \right) =\exp\left(-\frac{\|\mathbf{e}\|^2}{2\sigma^2}\right)
\end{equation}

Because the variance $\sigma^2$ is a static independent constant inherent to the LWE parameters, the statistical maximization of the log-likelihood function corresponds identically to the minimization of the residual squared norm: $\arg\max_{\mathbf{e}} \ln \mathcal{L}(\mathbf{e}) \equiv \arg\min_{\mathbf{e}} \|\mathbf{e}\|^2$.
\end{proof}


\section{NISQ Realization: Relaxation to an Approximate Model and Recovery of Exactness}
\label{sec:hardware}

While the mathematical formulations in Section \ref{sec:modeling} provide an exact mapping to a quadratic objective and represent a hardware-independent general model, evaluating them on current Noisy Intermediate-Scale Quantum (NISQ) devices requires overcoming physical precision limits. 
Therefore, we introduce an adaptive encoding strategy to successfully embed the problem onto the CIM. This approach makes a trade-off by introducing a continuous approximation, nevertheless, our hybrid framework provides intrinsic robustness to reconstructing the secret from heuristic hardware outputs.

\subsection{Approximate Encoding, Hardware Embedding and Result Recovery}
\subsubsection{Matrix Conditioning and CR-BNP}
\label{subsec:crbnp}
Instead of modeling the entire space as previous methods did, we partition the lattice into a \textit{frozen subspace} ($\mathcal{F}$) reserved for classical substitution, and an \textit{exploration subspace} ($\mathcal{E}$) exclusively mapped to the CIM. This partitioning strategy significantly reduces the number of variables in the QUBO formulation, making it highly feasible for current hardware capacities and precision limits.

While Section~\ref{sec:modeling} utilizes the standard Babai's Nearest Plane algorithm to provide an approximate closest lattice point, we introduce a modification here. Within the \textit{exploration subspace}, we deliberately defer the discrete rounding operation ($\lfloor \cdot \rceil$). Instead, we allow the modified algorithm to use continuous projection coefficients to sequentially compute the continuous coordinates ($\mathbf{y}_c$). We term this approach \textbf{Continuous Relaxed Babai's Nearest Plane (CR-BNP)}. 

Subsequently, we extract the fractional residuals ($\boldsymbol{\delta} = \mathbf{y}_c - \lfloor \mathbf{y}_c \rceil$) to drive the Adaptive Mixed-Radix Encoding and formulate our QUBO model. Once the discrete assignments for the \textit{exploration subspace} are resolved by the CIM, we apply the traditional discrete Babai's algorithm to determine the remaining coordinates in the \textit{frozen subspace} ($\mathcal{F}$). This continuous-discrete division not only mitigates premature rounding error propagation across the entire space but also explicitly extracts the continuous ambiguities required for the adaptive encoding detailed in the next step.

\subsubsection{Adaptive Mixed-Radix Dimensionality Reduction}
\label{subsec:encoder}
Recognizing that the fractional deviations $\delta_i$ quantify spatial proximity, we devise an \textbf{Adaptive Mixed-Radix Encoder} to bound the offset $\Delta \mathbf{y}$ introduced in Section~\ref{sec:modeling}. Because the CR-BNP projection already centers the search space near the true lattice point, the residual discrete offset is bounded. Thus, we constrain the search domain per dimension to the local integer neighborhood, $\Delta y_i \in \{-1, 0, 1\}$. 

Rather than blindly imposing expensive ternary expansions across all variables, the encoder operates exclusively within the \textit{exploration subspace}. Setting empirical continuous thresholds $\tau_{\text{low}}$ and $\tau_{\text{high}}$, it dynamically allocates bit-widths based on spatial ambiguity: dimensions with low ambiguity ($|\delta_i| < \tau_{\text{low}}$) are fixed as 0-bit constants; dimensions exhibiting strong directional tendencies ($|\delta_i| \ge \tau_{\text{high}}$) are pruned into 1-bit constrained Boolean biases; and only highly localized, ambiguous fractional limits ($\tau_{\text{low}} \le |\delta_i| < \tau_{\text{high}}$) retain a full 2-bit (ternary) encoding. This dynamic pruning systematically performs combinatorial dimensionality reduction, compressing the exponentially large raw configuration space into a highly compact logical domain. It also narrows the dynamic amplitude range of the resulting QUBO coefficients, making the model more amenable to hardware precision limits.

\subsubsection{Hardware Execution and Bounded Correction}
\label{subsec:hardwareexecution}
To maximize the utilization of available physical qubits (e.g., the 550-spin capacity of our target CIM), we employ a variable partitioning strategy. Let $N_{\text{log}}$ denote the total number of logical variables required for the QUBO model. We select a small number of pivot variables, denoted as $K$ (e.g., $K=5$ in Section~\ref{sec:experiments}), subject to the parallel embedding constraint $2^K \times (N_{\text{log}}-K) \le 550$. These $K$ variables are deterministically chosen from the 1-bit assigned dimensions by sorting their continuous fractional ambiguities $|\delta_i|$ in descending order. 

By exhaustively enumerating the $2^K$ possible Boolean assignments for these pivot variables, we inject each assignment as a constant bias into the QUBO formulation. The resulting $2^K$ independent sub-problems, each of size $N_{\text{log}} - K$, are then packed block-diagonally into a sparse global configuration. Prior to mapping, each sub-problem is individually linearly scaled by its maximum absolute coefficient amplitude and uniformly quantized. This strategy enables the entire ensemble to be mapped onto the CIM hardware in a single batched submission. 

Translating the continuous landscape into a discrete QUBO formulation introduces misalignments at integer boundaries. Confined by hardware precision limits, this misalignment can cause continuous pseudo-minima to incur rounding penalties upon forced discretization. To rescue these near-miss configurations upon retrieving the hardware solutions, we apply a classical bounded neighborhood corrector ($R \le 2$) over the full candidate vector to repair small deviations. 

Specifically, the $R=1$ search loops through all indices to test a single $\pm 1$ step mutation, while $R=2$ enumerates every pair of indices for combinatorial $\pm 1$ steps. Whenever the corrector actively mutates these 1 or 2 target dimensions, it deterministically shifts the continuous target projection for the remaining subspace and immediately invokes the discrete Babai's nearest plane algorithm to re-evaluate all preceding indices. Consequently, a single $\pm 1$ step on a trailing exploration coordinate can adaptively repair long propagated mis-roundings in the leading frozen coordinates.

\subsection{Complete Pipeline Summary}
\label{sec:pipeline_summary}
Algorithm~\ref{alg:qcbdd} summarizes the end-to-end pipeline integrating the three preceding components. The five phases proceed sequentially: algebraic isomorphism (Phase 1) and Gram-Schmidt projection (Phase 2) prepare the lattice geometry; CR-BNP (Phase 3, Section~\ref{subsec:crbnp}) extracts the continuous fractional ambiguity; the adaptive mixed-radix encoder (Phase 4, Section~\ref{subsec:encoder}) compresses the search into a logical-bit QUBO, with variable partitioning yielding $2^K$ block-diagonal sub-problems for one CIM submission; finally, post-processing verification and the $R\le 2$ bounded corrector (Phase 5, Section~\ref{subsec:hardwareexecution}) distinguish the true result and recover the secret. We defer the detailed process to Appendix~\ref{sec:appendix_pipeline}.

\begin{algorithm}[]
\caption{CIM-BDD: Systematic Execution Flow}
\label{alg:qcbdd}
\begin{algorithmic}[1]
\REQUIRE LWE instance $(\mathbf{A}, \mathbf{c})$ with $\mathbf{A} \in \mathbb{Z}_q^{m \times n}$, modulus $q$, standard deviation $\sigma$.
\ENSURE The true error vector $\mathbf{e}$ and secret $\mathbf{s}$.

\STATE \textbf{Phase 1: Algebraic Isomorphism \& Global Reduction}
\STATE Randomly select $m'$ rows of $\mathbf{A}$ such that the leading $n\times n$ 
       submatrix $\mathbf{A}_{\text{top}}$ is invertible modulo $q$ 
\STATE Compute transition matrix $\mathbf{T}_{\text{mat}} \leftarrow 
       \mathbf{A}_{\text{bot}}\mathbf{A}_{\text{top}}^{-1} \pmod q$.
\STATE Apply BKZ/LLL~\cite{Lenstra1982} on initial CVP basis to yield 
       integer basis $\mathbf{B}_0$.
\STATE Inherit GSO components: $\mathbf{B}_0 = \mathbf{B}^* \mathbf{U}_\mu$.

\STATE \textbf{Phase 2: Gram-Schmidt Projection}
\STATE Extract the corresponding $m'$-dimensional sub-target from $\mathbf{c}$, lift it to $\mathbf{c}'$, and project $\mathbf{c}'$ onto $\mathbf{B}^*$ to obtain full continuous coefficient vector $\boldsymbol{\alpha}$.
\STATE Isolate exploration subspace: extract $\boldsymbol{\alpha}_{\text{exp}}$ and $\mathbf{U}_{\mu, \text{exp}}$.

\STATE \textbf{Phase 3: Continuous Relaxed Babai's Nearest Plane (CR-BNP)}
\STATE Compute $\mathbf{y}_c$ via $\mathcal{O}(d_{\text{exp}}^2)$ Gaussian back-substitution: $\mathbf{U}_{\mu, \text{exp}} \mathbf{y}_c = \boldsymbol{\alpha}_{\text{exp}}$.

\STATE \textbf{Phase 4: Adaptive Encoding \& Subspace QUBO Gen.}
\STATE Extract continuous fractional ambiguity: $\boldsymbol{\delta} \leftarrow \mathbf{y}_c - \lfloor \mathbf{y}_c \rceil$.
\STATE Map continuous $\boldsymbol{\delta}$ into logical bits using 0/1/2-bit mappings and thresholds $\tau_{\text{low}}, \tau_{\text{high}}$.
\STATE Formulate affine transformation: $\tilde{\mathbf{y}}_{\text{exp}} \leftarrow \mathbf{y}_{\text{base}} + \mathbf{T}_{\text{enc}}\mathbf{z}$.
\STATE Formulate QUBO matrix $E_{\text{proj}}(\mathbf{z})$ purely within $\mathcal{F}^\perp$.
\STATE Apply variable partitioning (exhaustive enumeration of $K$ pivot bits), pack independent instances block-diagonally, and map to CIM.

\STATE \textbf{Phase 5: Verification \& Recovery}
\FOR{each partition branch hardware sample $\mathbf{z}_{\text{hw}}$ extracted from a single CIM submission}
    \STATE Reconstruct physical displacement $\mathbf{y}_{\text{exp}} \leftarrow \mathbf{y}_{\text{base}} + \mathbf{T}_{\text{enc}}\mathbf{z}_{\text{hw}}$.
    \STATE Compute frozen coordinates $\mathbf{y}_{\text{froz}}$ using Discrete Babai's nearest plane algorithm.
    \STATE Let candidate vector $\mathbf{y}_{\text{cand}} = [\mathbf{y}_{\text{froz}}^\top, \mathbf{y}_{\text{exp}}^\top]^\top$.
    \STATE Evaluate global discrete energy $E \leftarrow \|\mathbf{c}' - \mathbf{B}_0 \mathbf{y}_{\text{cand}}\|^2$.
    \IF{$E > T_{\text{early}}$}
        \STATE Perform bounded neighborhood corrector ($R \le 2$) over all coordinates of $\mathbf{y}_{\text{cand}}$, each coordinate correction ($\pm 1$ step) invokes Babai's nearest plane algorithm on front indices. Update $E$ and $\mathbf{y}_{\text{cand}}$.
    \ENDIF
    \IF{$E \le T_{\text{early}}$}
        \STATE \textbf{Output:} Recover $\mathbf{s} \leftarrow (\mathbf{A}_{\text{top}}^{-1} (\mathbf{B}_0 \mathbf{y}_{\text{cand}})_{\text{top}}) \pmod q$.
        \RETURN $\mathbf{s}$, and the global error $\mathbf{e} = (\mathbf{c} - \mathbf{A}\mathbf{s}) \bmod q$
    \ENDIF
\ENDFOR
\end{algorithmic}
\end{algorithm}

\subsection{Sources of Inaccuracy and the Recovery Method}
The pipeline does not solve the LWE instance exactly, it solves the
$d_{\text{exp}}$-dimensional exploration QUBO (near-)exactly. Sources of inexactness include: the heuristic $0/1/2$-bit encoding window omitting the true offset, continuous-integer misalignments altering the QUBO landscape, limited hardware input precision, and frozen-subspace Babai mis-rounding.

The recovery layer is designed to rescue these losses. Each CIM submission returns an ensemble of low-energy configurations at distinct energy levels. Whenever the true point coincides with one of them or lies within the $R\le 2$ neighborhood of one, the bounded corrector restores it and the $T_{\text{early}}$ certificate confirms exactness. Section~\ref{sec:experiments} validates this behavior empirically.

\section{Theoretical Analysis}
\label{sec:theory}

\subsection{Orthogonal Decomposition and the Role of the Frozen Subspace}

We first show that the global squared error splits across the two
subspaces. Let $\mathcal{F}$ denote the frozen subspace handled by classical
substitution and $\mathcal{F}^\perp$ the exploration subspace mapped to QUBO. Any global error vector $\mathbf{e}'$ decomposes uniquely into a component
parallel to $\mathcal{F}$ (denoted $\mathbf{e}_{\text{froz}} =
\pi_{\mathcal{F}}(\mathbf{e}')$) and one parallel to $\mathcal{F}^\perp$ (denoted $\mathbf{e}_{\text{proj}}(\mathbf{z}) = \pi_{\mathcal{F}^\perp}(\mathbf{e}')$, which is explicitly parameterized by the logical variables $\mathbf{z}$):
\begin{equation}
    \mathbf{e}' = \mathbf{e}_{\text{froz}} + \mathbf{e}_{\text{proj}}(\mathbf{z}).
\end{equation}
By orthogonality the squared Euclidean norm separates additively:
\begin{equation}
    \|\mathbf{e}'\|^2 = \|\mathbf{e}_{\text{froz}}\|^2
    + \|\mathbf{e}_{\text{proj}}(\mathbf{z})\|^2. \label{eq:pythagorean}
\end{equation}
Eq.~\eqref{eq:pythagorean} is the structural basis for splitting the work
between hardware and classical post-processing. The second term,
$\|\mathbf{e}_{\text{proj}}(\mathbf{z})\|^2$, is confined to
the exploration subspace, minimizing the QUBO energy $E(\mathbf{z})$ minimizes it over the encoded domain. Conversely, for a
fixed hardware assignment $\mathbf{z}_{\text{opt}}$, minimizing the total error
$\|\mathbf{e}'\|^2$ reduces to minimizing the first term
$\|\mathbf{e}_{\text{froz}}\|^2$, a standard lower-dimensional CVP over the
frozen basis $\mathbf{B}_{\text{froz}}$ that the classical post-processor
resolves by Babai's nearest-plane algorithm.

Eq.~\eqref{eq:pythagorean} holds as a continuous identity, but our modeling imposes discrete integer constraints, which may couple the two subspaces. Therefore,  minimizing the total error does not strictly decouple into independent minimizations of the two terms. This is one source of the inaccuracy of our modeling.

\subsection{Lattice Geometry and the Choice of Subdimension}
\label{subsec:justification}

\begin{lemma}[Length of the Last Gram-Schmidt Vector]
\label{lem:gs_length}
Let $\Lambda$ be an $m'$-dimensional $q$-ary lattice with volume $q^{m'-n}$, corresponding to an LWE instance with secret dimension $n$. Under the Geometric Series Assumption (GSA), if the lattice is reduced by an algorithm with root-Hermite factor $\gamma_H$, the length of the last Gram--Schmidt vector is estimated as:
\begin{equation}
    \|\mathbf{b}_{m'}^*\| \approx q^{1 - \frac{n}{m'}} \cdot {\gamma_H}^{-(m'-1)}.
    \label{eq:tail_gs}
\end{equation}
\end{lemma}

\begin{proof}
Under the GSA heuristic, the lengths of the Gram-Schmidt vectors decay geometrically, i.e.,
$\|\mathbf{b}_i^*\| \approx \rho^{\,i-1}\|\mathbf{b}_1^*\|$ with $\rho = {\gamma_H}^{-2}$.
Since the lattice volume equals the product of these lengths, we have:
\begin{equation*}
    \operatorname{Vol}(\Lambda) = \prod_{i=1}^{m'} \|\mathbf{b}_i^*\| \approx \|\mathbf{b}_1^*\|^{m'} \rho^{m'(m'-1)/2} = q^{m'-n}.
\end{equation*}
Solving for the first vector yields $\|\mathbf{b}_1^*\| \approx q^{1-n/m'} \rho^{-(m'-1)/2}$. Substituting into
$\|\mathbf{b}_{m'}^*\| \approx \rho^{\,m'-1}\|\mathbf{b}_1^*\|$ and using
$\rho^{1/2}={\gamma_H}^{-1}$ gives Eq.~\eqref{eq:tail_gs}.
\end{proof}

Setting the derivative of Eq.~\eqref{eq:tail_gs} to zero yields an idealized
GSA-optimal subdimension of $m'=152$ for our instance
(Figure~\ref{fig:gs_length}). We chose the sub-optimal point $m'=88$, which actually is a harder operating point, mainly for \emph{demonstration and stress-testing}, not to maximize the recovery rate. A real attacker would tune toward larger
$m'$ with a correspondingly stronger block size $\beta$; our choice of $m'=88$ is a conservative operating point that exposes the
failure-and-recovery behaviour of the pipeline, with a shorter time of BKZ reduction because larger $m'$ slows the BKZ, which will be a problem when considering larger instances.

\begin{figure}[h]
\centering
\includegraphics[scale=0.3]{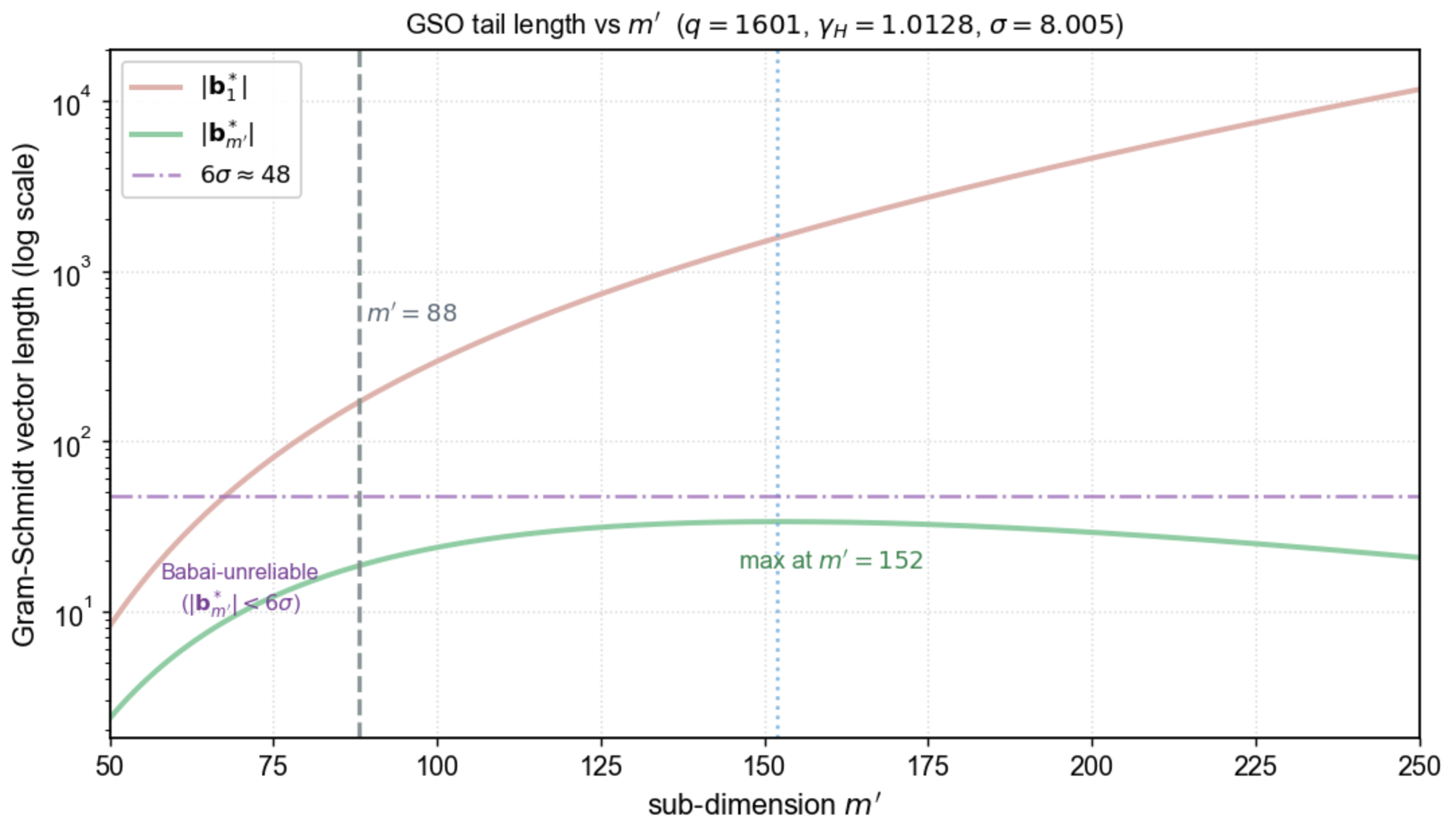}
\caption{Trend of the theoretical \( \|\mathbf{b}_1^*\| \) and \( \|\mathbf{b}_{m'}^*\| \) as functions of sub-dimension \( m' \) for the \( n = 40 \) instance (\( q = 1601, \gamma_H = 1.0128 \ (\text{BKZ-20}), \sigma = 8.005 \)). The dash-dotted purple line at \( 6\sigma \approx 48 \) is the Babai hard-decision failure region threshold, the grey vertical line at \( m' = 88 \) marks the point used in the experiments, the analytic maximum of \( \|\mathbf{b}_{m'}^*\| \) occurs at \( m' = 152 \).}
\label{fig:gs_length}
\end{figure}

\subsection{Decoding Limits and Exploration Subspace Reliability}
\label{subsubsec:decoding_limit}

To quantify the decoding limits exclusively within the hardware-mapped exploration subspace, we first evaluate the continuous coordinate error at the last dimension using empirical values instead of the theoretical ones in Section~\ref{subsec:justification}.
 Because its coordinate error avoids any accumulated covariance from other variables, evaluating this dimension provides a clean baseline for our estimation.

\begin{lemma}[Standard Deviation of the Continuous Coordinate Error at the Last Dimension]
\label{lem:variance_last}
In the CR-BNP framework, let the continuous coordinates $\mathbf{y}_c$ solve $\mathbf{U}_{\mu,\text{exp}}\mathbf{y}_c = \boldsymbol{\alpha}_{\text{exp}}$ with the true error $\mathbf{e}\sim \mathcal{D}_{\mathbb{Z}^{m'},\sigma}$. The standard deviation of the coordinate error at the last dimension $m'$ evaluates to:
\begin{equation}
    \sigma_{\text{last}} = \frac{\sigma}{\|\mathbf{b}_{m'}^*\|}.
\end{equation}
\end{lemma}

\begin{proof}
The continuous projection coefficient is
$\alpha_i = \langle \mathbf{c}', \mathbf{b}_i^*\rangle / \|\mathbf{b}_i^*\|^2$.
Writing $\mathbf{c}' = \mathbf{B}^*\mathbf{U}_\mu\mathbf{y}_{\text{true}} + \mathbf{e}$
gives, on the exploration block,
$\boldsymbol{\alpha}_{\text{exp}} = \mathbf{U}_{\mu,\text{exp}}\mathbf{y}_{\text{true},\text{exp}} + \mathbf{e}_{\text{proj},\text{exp}}$
with $e_{\text{proj},i} = \langle \mathbf{e}, \mathbf{b}_i^*\rangle / \|\mathbf{b}_i^*\|^2$.
Hence the coordinate error is
$\boldsymbol{\epsilon}_y = \mathbf{y}_{c,\text{exp}} - \mathbf{y}_{\text{true},\text{exp}}
= \mathbf{U}_{\mu,\text{exp}}^{-1}\mathbf{e}_{\text{proj},\text{exp}}$. Since
$\mathbf{U}_{\mu,\text{exp}}$ is upper unitriangular, so is its inverse, whose
last row is a single $1$ on the diagonal, thus
$\epsilon_{y,m'} = e_{\text{proj},m'} = \langle \mathbf{e}, \mathbf{b}_{m'}^*\rangle / \|\mathbf{b}_{m'}^*\|^2$.
Projecting $\mathbf{e}\sim \mathcal{D}_{\mathbb{Z}^{m'},\sigma}$ onto the
fixed vector $\mathbf{b}_{m'}^*$ gives
\begin{equation}
    \operatorname{Var}(\epsilon_{y,m'})
    = \frac{\operatorname{Var}(\langle \mathbf{e}, \mathbf{b}_{m'}^*\rangle)}{\|\mathbf{b}_{m'}^*\|^4}
    = \frac{\sigma^2\|\mathbf{b}_{m'}^*\|^2}{\|\mathbf{b}_{m'}^*\|^4}
    = \frac{\sigma^2}{\|\mathbf{b}_{m'}^*\|^2},
\end{equation}
and $\sigma_{\text{last}} = \sigma/\|\mathbf{b}_{m'}^*\|$ follows.
\end{proof}

We now estimate how ternary widening removes the per-coordinate bottleneck in the exploration subspace.

\begin{claim}[Last Coordinate Reliability via Ternary Widening]
\label{thm:qubo_reliability_single}
Restricting our analysis to the exploration subspace, we take a specific randomized row-selection seed for the $n=40$ configuration ($m'=88$) as a representative empirical example. Here, the Gram-Schmidt tail degrades to $\|\mathbf{b}_{88}^*\| \approx 15.531$ (slightly below the idealized GSA prediction). By Lemma~\ref{lem:variance_last}, the continuous coordinate error at this last dimension exhibits a standard deviation $\sigma_{\text{last}} \approx 8.005/15.531 \approx 0.515$. 

In this exploration region, a classical single-pass nearest-plane rounding succeeds at this last coordinate only if the fractional error lies in $[-0.5, 0.5]$, corresponding to $\Pr(|Z| \le 0.5/0.515) = 2\Phi(0.971) - 1 \approx 0.668$. 

Assuming the last coordinate is assigned a full ternary QUBO encoding ($\Delta y_{88} \in \{-1,0,1\}$), the interval widens from $\pm 0.5$ to $\pm 1.5$. This expansion restores the per-coordinate reliability at the tail:
\[
    p_{\text{last}} = \Pr\!\left(|Z| \le \tfrac{1.5}{0.515}\right) = 2\Phi(2.913) - 1 \approx 0.9964.
\]
\end{claim}

\begin{claim}[Cumulative Coverage of the Exploration Subspace]
\label{thm:qubo_coverage_cumulative}
To estimate the cumulative coverage within the $d_{\text{exp}}=12$ exploration subspace, we formulate a probabilistic projection under three assumptions: (1) per-coordinate statistical independence; (2) the last-dimension probability $p_{\text{last}}$ uniformly applies across all exploration dimensions; and (3) every coordinate is universally granted a full ternary encoding. Under these conditions, the theoretical cumulative coverage within the exploration subspace evaluates to:
\[
    P_{\text{explore}} \approx (p_{\text{last}})^{12} \approx (0.9964)^{12} \approx 0.957.
\]
\end{claim}

Note that $P_{\text{explore}} \approx 95.7\%$ serves as an estimation for the actual exploration subspace coverage, because actually the coordinates are not statistically independent, the standard deviations for dimensions in the exploration space are not equal to $\sigma_{\text{last}}$ generally, and our adaptive encoder does not apply full ternary widths to all dimensions.

\subsection{Early-Stopping Certificate and DLWE Distinguishing Advantage}
\label{subsec:early_stopping}

We establish a mathematical early-stopping threshold $T_{\text{early}}$ that functions as a one-sided cryptographic certificate. We formulate a distinguisher $\mathcal{D}$ that evaluates the residual energy of the recovered lattice point to bound the Decision-LWE advantage.

\begin{theorem}[One-Sided Distinguisher and DLWE Advantage]
\label{thm:early_stopping}
Let $\mathcal{D}$ be a Decision-LWE distinguisher that executes the encoded CIM-BDD solver. On input challenge $(\mathbf{A}, \mathbf{c})$, if the solver recovers a candidate vector $\mathbf{y}_{\text{cand}}$ satisfying the residual energy bound $E = \|\mathbf{c}'-\mathbf{B}_0\mathbf{y}_{\text{cand}}\|^2 \le T_{\text{early}}$, the distinguisher outputs $1$ (guessing \emph{LWE}). Otherwise, it defaults to outputting $0$ (guessing \emph{Uniform}). The threshold is analytically defined as:
\begin{equation}
    T_{\text{early}} = m'\sigma^2 + 4\sigma^2\sqrt{2m'}. \label{eq:tearly}
\end{equation}
Let $p_{\text{succ}}$ denote the solver's operational probability of correctly locating the true LWE error vector. The standard distinguishing advantage $\mathrm{Adv}^{\mathrm{DLWE}}(\mathcal{D}) = |\Pr[\mathcal{D}(\mathrm{LWE}) = 1] - \Pr[\mathcal{D}(\mathrm{Uniform}) = 1]|$ is lower-bounded by:
\begin{equation}
    \mathrm{Adv}^{\mathrm{DLWE}}(\mathcal{D}) \ge p_{\text{succ}} \cdot (1 - \alpha_{\text{tail}}) - \epsilon_{\text{vol}},
\end{equation}
where $\alpha_{\text{tail}}$ is the statistical right-tail probability of a $\chi^2_{m'}$ distribution evaluated at $m'+4\sqrt{2m'}$, and $\epsilon_{\text{vol}} = \operatorname{Vol}\!\big(\mathcal{B}_{m'}(\sqrt{T_{\text{early}}})\big)/\det(\Lambda_q)$ is the volumetric false-acceptance bound derived from the lattice Gaussian Heuristic.
\end{theorem}

\begin{proof}
\emph{(i) Evaluation on True LWE Instances ($\Pr[\mathcal{D}(\mathrm{LWE}) = 1]$):} 
Normalizing $X_i=e_i/\sigma$, the true-point energy is
$E=\|\mathbf{e}\|^2=\sigma^2\sum_{i=1}^{m'}X_i^2$, with
$\sum_i X_i^2\sim\chi^2_{m'}$ (mean $m'$, variance $2m'$). Hence
$\mathbb{E}[E]=m'\sigma^2$ and $\operatorname{Var}(E)=2m'\sigma^4$. Placing the
threshold four standard deviations above the mean gives
Eq.~\eqref{eq:tearly}, i.e.\ $T_{\text{early}}/\sigma^2=m'+4\sqrt{2m'}$. The mathematical probability that the true planted error stochastically violates this generous threshold is bounded by the analytic tail $\alpha_{\text{tail}} = \Pr[\chi^2_{m'} > m'+4\sqrt{2m'}]$. Since the algorithm empirically locates this true lattice point with probability $p_{\text{succ}}$, the joint probability of outputting $1$ is well approximated by $p_{\text{succ}} \cdot (1 - \alpha_{\text{tail}})$.

\emph{(ii) Evaluation on Uniform Instances ($\Pr[\mathcal{D}(\mathrm{Uniform}) = 1]$):} 
A uniform random target $\mathbf{u} \sim \mathcal{U}(\mathbb{Z}_q^{m'})$ contains no planted error. The distinguisher outputs $1$ only if a wrong lattice point accidentally falls within the Euclidean ball $\mathcal{B}_{m'}(\sqrt{T_{\text{early}}})$ centered at the target. By the Gaussian Heuristic, the expected number of such lattice points asymptotically upper-bounds this probability: $\Pr[\mathcal{D}(\text{Uniform}) = 1] \le \epsilon_{\text{vol}} = \operatorname{Vol}(\mathcal{B}_{m'}(\sqrt{T_{\text{early}}}))/\det(\Lambda_q)$.

\emph{(iii) Distinguishing Advantage:} 
Substituting these bounds yields
\begin{align*}
    \mathrm{Adv}^{\mathrm{DLWE}}(\mathcal{D}) &= \Pr[\mathcal{D}(\text{LWE}) = 1] - \Pr[\mathcal{D}(\text{Uniform}) = 1] \\
    &\ge p_{\text{succ}} \cdot (1 - \alpha_{\text{tail}}) - \epsilon_{\text{vol}}.
\end{align*}
\end{proof}

This early-stopping threshold $T_{\text{early}}$ serves as a one-sided statistical certificate: crossing below $T_{\text{early}}$ certifies
the secret, whereas failing to cross is uninformative as to whether the target
was LWE or uniform. It gives a search-to-decision distinguisher, whose advantage approximately equals the search success probability. For Search-LWE, the certificate accelerates verification that could also be done by recomputing $\mathbf{s}$;
for Decision-LWE, a uniform instance has \emph{no} secret to recompute, so the threshold-crossing
event is the only available distinguisher. 


\section{Experiments and Evaluation}
\label{sec:experiments}

All experiments are conducted on CPQC-550 at the Coherent Ising Machine
platform through the Kaiwu SDK~\cite{kaiwu_sdk_2022}, with all instances drawn
from the TU Darmstadt LWE Challenge~\cite{lattice_challenge}. 

\subsection{Toy Instance Validation}
\label{sec:n5_profiler}
Although the toy instances could be solved directly using the modeling in Section~\ref{sec:modeling}, to validate the pipeline presented in Section~\ref{sec:hardware} --- the CR-BNP
projection, the adaptive mixed-radix encoder, the block-diagonal variable
partitioning, and the bounded corrector ($R\le 2$) --- we route the $n=5$ instance
($m'=12$, $q=29$, $\sigma\approx 0.29$) through every stage with an exploration
subspace of $d_{\text{exp}}=8$. Setting $\tau_{\text{low}}=0.05$ and $\tau_{\text{high}}=0.20$, the encoder collapses the search space into $4$
logical bits ($0$-bit:$6$, $2$-bit:$2$), with a $\chi^2$ early-stopping
threshold $T_{\text{early}}\approx 2.7$ from
$m'\sigma^2 + 4\sigma^2\sqrt{2m'} = 12(0.0841) + 4(0.0841)\sqrt{24} \approx 2.6562$.

A single CIM submission returns four distinct results
(Table~\ref{tab:n5_profiler}). Samples \#1 and
\#2 (lowest $E_{\text{QUBO}}=-434$) are the discrete ground state directly
($E_{\text{total}}=1$), \#3 and \#4 are pseudo-minima whose true lattice energy
reaches $E_{\text{total}}=73$ and $162$, and the bounded corrector pays this
gap at $R{=}1$ and $R{=}2$, all landing on the same ground state
$E{=}1 < T_{\text{early}}=2.7$. The pipeline recovers the secret
$\mathbf{s}=(7,27,14,23,26)$. The same submission yields both
direct ground states and pseudo-minima: the framework requires
only that the ground state lie within the $R\le 2$ neighborhood of \emph{some} returned candidate, a weaker requirement than other
schemes that treat QUBO bit-strings as literal ground states.

Note that the two-bit map represents the ternary step set $\{-1, 0, +1\}$ with four
bit-strings, so the central value is doubly encoded. This is a deliberately
redundant, sign-symmetric encoding: it keeps the search window centered from
the Babai's algorithm at the cost of one redundant code-point, and it induces a
ground-state degeneracy. For instance, samples \#1 and \#2 in
Table~\ref{tab:n5_profiler} are distinct bit-strings of equal energy.

\begin{table}[ht]
\centering
\caption{Results of \texttt{LWE\_5\_010} ($q=29$, $\sigma\approx0.29$). Four distinct results from a single CIM submission. $E_{\text{proj}}$ is the continuous-relaxation energy (what the QUBO minimizes), $E_{\text{QUBO}}$ is the Ising objective (an affine image of $E_{\text{proj}}$, not a lattice energy), $E_{\text{total}}$ is the true integer lattice energy.  
Samples \#1--\#2 (lowest $E_{\text{QUBO}}$) are the discrete ground state directly, \#3--\#4 are the pseudo-minima ($E_{\text{total}}=73,162$) which the bounded corrector ($R\le2$) rescues, all recovering the same ground state $E{=}1$, below the $\chi^2$ certificate $T_{\text{early}}=2.7$.}
\label{tab:n5_profiler}
\smallskip
\renewcommand{\arraystretch}{1.1}
\setlength{\tabcolsep}{6pt}
{
\begin{tabular}{ccccccc}
\toprule
\textbf{Sample} & \textbf{Source} & \textbf{Bitstring} & $E_{\text{QUBO}}$ & $E_{\text{proj}}$ & $E_{\text{total}}$ & \textbf{Corrector} \\
\midrule
\#1 & \texttt{Res:0} & $[0,1,0,1]$ & $-434$ & $0.76$  & $\mathbf{1}$ & $R{=}0$ (direct) \\
\#2 & \texttt{Res:1} & $[1,0,1,0]$ & $-434$ & $0.76$  & $\mathbf{1}$ & $R{=}0$ (direct) \\
\#3 & \texttt{Res:2} & $[1,0,0,0]$ & $-216$ & $59.92$ & $73$ & $R{=}1$ on $\text{dim}_{07}\!\to\!E{=}1$ \\
\#4 & \texttt{Res:3} & $[1,1,1,1]$ & $+116$ & $150.06$ & $162$ & $R{=}2$ on $\text{dim}_{07},\dim_{04}\!\to\!E{=}1$ \\
\bottomrule
\end{tabular}}
\end{table}
\emph{Remark.} 
This instance (\texttt{LWE\_5\_010}) could be solved by the modeling in Section~\ref{sec:modeling} alone, since classical Babai's nearest-plane (only use LLL) already finds the solution ($\Delta\mathbf{y}=\mathbf{0}$). We use it instead to trace the Section~\ref{sec:hardware} procedure end-to-end, the exploration-subspace division also illustrates the encoder's qubit reduction relative to modeling the full instance. Preprocessing here requires only LLL reduction rather than BKZ-$20$.

\subsection{Validation at $n=40$ (\texttt{LWE\_40\_005}): Multi-Seed Study}
\label{subsec:n40_results}

At $n=40$ with $m'=88$ and BKZ-$20$ (or even BKZ-$30$) reduction, a single pass
of Babai's nearest-plane algorithm fails to recover the correct lattice point,
yielding baseline residuals of order $6\times10^4$, several-fold above the $\chi^2$ threshold $T_{\text{early}}=9040$. Holding the
\texttt{LWE\_40\_005} instance and all algorithmic parameters fixed, we vary
only the random seed that selects an $\mathbf{A}_{\text{top}}$ invertible modulo
$q$, yielding $11$ row partitions ($\text{seed}\in\{42,1,2,\ldots,10\}$) of the
\emph{same} underlying LWE problem. Each seed is evaluated under two reduction
strengths: \textbf{Baseline} single-pass BKZ-$20$ and \textbf{Progressive}
BKZ-$20\!\to\!30$. For each
seed, CR-BNP isolates $d_{\text{exp}}=12$ exploration dimensions, the adaptive
encoder with $\tau_{\text{low}}=0.05$ and $\tau_{\text{high}}=0.20$ compresses them into $10$--$19$ logical bits depending on the empirical
$\delta$ distribution, and variable partitioning splits the instance into $32$
block-diagonal branches aggregated into one submission.

Table~\ref{tab:n40_multiseed} summarizes the results.
Baseline BKZ-$20$ recovers the secret on $2$ of $11$ seeds ($18\%$; $3$ hits in total); progressive BKZ-$20\!\to\!30$ doubles this to $4$ of $11$ ($36\%$; $21$ hits). On every successful seed, the residual energy follows a consistent descent: $E_{\text{Babai}} \sim 4\times 10^4\text{--}6\times10^4$ (single-pass nearest-plane, several-fold above $T_{\text{early}}$) $\to E_{R=0} \sim 4\times 10^4\text{--}7\times 10^4$ (CIM continuous pseudo-minimum evaluated on the discrete lattice) $\to E_{\text{final}} \in [5{,}270, 6{,}509]$ (post-corrector discrete ground state), spanning roughly one order of magnitude. Specifically, because each randomized minimal row selection extracts a distinct $m'$-dimensional subset of the global discrete Gaussian error, the true ground-state energies naturally fluctuate per seed. The empirically resolved $E_{\text{final}}$ values tightly cluster around the theoretically expected residual energy $\mu_E = m'\sigma^2 \approx 5{,}639$ and successfully get through the early-stopping threshold $T_{\text{early}} = 9{,}040$, confirming an exact cryptographic lattice point recovery.

The estimated single-pass success probability for each seed could be calculated as $P_{\text{comp}} = P_{\text{froz}} \times P_{\text{explore}}$. As established in Claim~\ref{thm:qubo_coverage_cumulative}, the ternary widening inherently pushes $P_{\text{explore}} \approx 95.7\%$, isolating the system's mathematical bottleneck to the frozen subspace limit: $P_{\text{froz}} = \prod_{i=1}^{76} \left[ 2\Phi\left(\frac{\|\mathbf{b}_i^*\|/2}{\sigma}\right) - 1 \right]$. 
For baseline BKZ-20, the theoretical mean $\overline{P}_{\text{comp}} = 18.0\%$ agrees with our empirical recovery of $18.18\%$ ($2/11$). 
Under progressive BKZ-20$\to$30 reduction, the empirical recovery rises to $36.36\%$ ($4/11$), higher than the $26.5\%$ mathematical expectation. We read the two as consistent in magnitude rather than as a
statistically significant gap; the modest excess is the expected signature of the corrector, while under the weak BKZ-20 baseline, the corrector doesn't help much due to the poorly reduced basis.

One specific seed (\texttt{seed=4}) succeeds under baseline BKZ-$20$ but fails under BKZ-$20\!\to\!30$. This occurs because the stronger reduction shifts the $\delta$-encoding distribution. So BKZ-$30$ is beneficial for most seeds but occasionally adversarial for an individual instance. 

The most informative single run is \texttt{seed=7}, the only seed that recovers the secret under \emph{both} BKZ configurations, at the identical post-corrector energy $E_{\text{final}}=5{,}270$: the same cryptographic ground state reached via two different lattice geometries. 

For the hardware time, we observed the single batched $t_{\text{CIM}}$ stays in the millisecond regime ($0.24$--$17.36$\,ms across the $22$ runs), and did not observe a correlation between time and block size. Also, $t_{\text{CIM}}$ carries no information about recovery: the fastest run overall ($0.24$\,ms, seed~$10$/baseline) fails, a comparably fast run ($0.39$\,ms, seed~$5$/progressive) succeeds, and the slowest ($17.36$\,ms, seed~$42$/baseline) fails.

Recovery is also stable under small changes of the exploration dimension. Repeating the baseline study with $d_{\text{exp}}=11$ (all else fixed) gives an essentially identical outcome, with the same two seeds succeeding (\texttt{seed=4}, \texttt{seed=7}, $18\%$) at the same post-corrector energies $E_{\text{final}}=6{,}509$ and $5{,}270$, recovering the same secret and agreeing on the failing seeds. 

\begin{table}[ht]
\centering
\caption{Multi-seed robustness study on \texttt{LWE\_40\_005} ($n=40$, $m'=88$, $q=1601$, $T_{\text{early}}=9040$). Each row is an independent run differing only in the random seed for row selection. For each reduction we report $N_{\text{log}}$, the number of logical qubits in QUBO; $t_{\text{CIM}}$, the wall-clock time of a single CIM submission; the per-seed analytic success rate $P_{\text{comp}}$; and the final outcome. For the baseline we also list $E_{\text{final}}$, the post-corrector residual energy; for the progressive run, $\delta$-valid, whether the encoder covers the true solution within $\{0,\pm 1\}$ steps per dimension. The ensemble mean of $P_{\text{comp}}$ predicts the recovery rate.
}
\label{tab:n40_multiseed}
\smallskip
\renewcommand{\arraystretch}{1.1}
\setlength{\tabcolsep}{3pt}
{\small
\resizebox{\textwidth}{!}{%
\begin{tabular}{c|ccccc|ccccc}
\toprule
& \multicolumn{5}{c|}{\textbf{Baseline BKZ-}$\mathbf{20}$} & \multicolumn{5}{c}{\textbf{Progressive BKZ-}$\mathbf{20\!\to\!30}$} \\
\textbf{seed} & $N_{\text{log}}$ & $E_{\text{final}}$ & $t_{\text{CIM}}$ (ms) & $P_{\text{comp}}$ & \textbf{Outcome} & $\delta$-valid & $N_{\text{log}}$ & $t_{\text{CIM}}$ (ms) & $P_{\text{comp}}$ & \textbf{Outcome} \\
\midrule
$42$ & $15$ & $27{,}070$ & $17.36$ & $0.189$ & fail & $\checkmark$ & $16$ & $9.30$ & $0.278$ & \textcolor{PassGreen}{break (5 hits)} \\
$1$  & $17$ & $26{,}944$ & $0.80$  & $0.169$ & fail & $\checkmark$ & $17$ & $8.82$ & $0.267$ & fail \\
$2$  & $15$ & $24{,}880$ & $9.99$  & $0.181$ & fail & $\times$    & $13$ & $0.93$ & $0.306$ & \textcolor{PassGreen}{break ($E\!=\!5{,}778$)} \\
$3$  & $12$ & $27{,}537$ & $0.60$  & $0.184$ & fail & $\times$    & $17$ & $8.26$ & $0.260$ & fail \\
$4$  & $15$ & $\mathbf{6{,}509}$  & $0.41$ & $0.170$ & \textcolor{PassGreen}{break} & $\times$    & $15$ & $3.08$ & $0.272$ & fail \\
$5$  & $17$ & $26{,}712$ & $6.83$ & $0.200$ & fail & $\checkmark$ & $10$ & $0.39$ & $0.248$ & \textcolor{PassGreen}{break (13 hits)} \\
$6$  & $16$ & $27{,}938$ & $5.28$ & $0.179$ & fail & $\checkmark$ & $19$ & $9.62$ & $0.288$ & fail \\
$7$  & $15$ & $\mathbf{5{,}270}$  & $9.42$ & $0.191$ & \textcolor{PassGreen}{break (2 hits)} & $\checkmark$ & $14$ & $0.56$ & $0.224$ & \textcolor{PassGreen}{break (2 hits)} \\
$8$  & $15$ & $28{,}885$ & $10.22$ & $0.178$ & fail & $\times$    & $12$ & $0.87$ & $0.267$ & fail \\
$9$  & $11$ & $25{,}502$ & $0.91$ & $0.151$ & fail & $\times$    & $12$ & $0.77$ & $0.247$ & fail \\
$10$ & $13$ & $27{,}953$ & $0.24$ & $0.188$ & fail & $\times$    & $11$ & $0.55$ & $0.257$ & fail \\
\midrule
\textbf{Summary} & & & & $\overline{P}_{\text{comp}}{=}\mathbf{0.180}$ & $2/11{=}\mathbf{18\%}$ & & & & $\overline{P}_{\text{comp}}{=}\mathbf{0.265}$ & $4/11{=}\mathbf{36\%}$ \\
\bottomrule
\end{tabular}}}
\end{table}


\subsection{Decision-LWE (DLWE) Validation}
\label{subsec:dlwe}
We now instantiate the theoretical bounds of Theorem~\ref{thm:early_stopping} using the specific $n=40$ configuration ($m'=88, \sigma=8.005$). The early-stopping certificate threshold evaluates to $T_{\text{early}}=9040$, the theoretical right-tail probability is $\alpha_{\text{tail}} \approx 2.85\times10^{-4}$, and the theoretical Gaussian Heuristic volume limits the volumetric false-acceptance bound to $\epsilon_{\text{vol}} \approx 5\times10^{-13}$.

To validate the distinguishing advantage $\mathrm{Adv}^{\mathrm{DLWE}}(\mathcal{D})$, we evaluate the empirical acceptance probabilities based on the event $\{E \le T_{\text{early}}\}$:

As demonstrated in Section~\ref{subsec:n40_results}, successfully solved LWE seeds drive the residual energy below  $T_{\text{early}}$ (e.g., $E_{\text{final}} \in [5{,}270, 6{,}509]$). However, when the solver fails, it yields a residual energy overlapping the $\sim\!2.5\text{--}2.9\times10^4$ range. Because these failed instances miss the threshold, the distinguisher defaults to $0$. Consequently, the empirical probability of outputting $1$ on an LWE target is bounded by the algorithmic recovery rate: $\Pr[\mathcal{D}(\text{LWE})=1] = p_{\text{succ}} \cdot (1 - \alpha_{\text{tail}}) \approx p_{\text{succ}}$.

To validate the null hypothesis, we generated $N=220$ independent uniform-random targets $\mathbf{c}_{\text{rand}}\sim\mathcal{U}(\mathbb{Z}_q^{m'})$, then evaluated them using a classical lattice enumeration solver. Across all targets, zero instances crossed $T_{\text{early}}$ (the minimum global energy observed was $28{,}564$, far above the threshold, as plotted in Figure~\ref{fig:dlwe}). The
single hardware (CIM) control run is consistent with this ($E_{\text{rand}}=28564$). This confirms the theoretically negligible estimation ($\epsilon_{\text{vol}} \approx 5\times10^{-13}$). 

Substituting the empirical boundaries into Theorem~\ref{thm:early_stopping}, the DLWE distinguishing advantage of the CIM-BDD framework evaluates to the algorithmic recovery rate:
\begin{align*}
\mathrm{Adv}^{\mathrm{DLWE}} &= \big|\Pr[\mathcal{D}(\text{LWE})=1] - \Pr[\mathcal{D}(\text{Uniform})=1]\big| \\
&\approx | p_{\text{succ}} - 0 | = p_{\text{succ}}.
\end{align*}
This yields an absolute single-shot distinguishing advantage of $0.18$ (under baseline BKZ-20) and $0.36$ (under progressive BKZ-20$\to$30). Amortizing over $k$ independent row-selection seeds, the advantage robustly amplifies to $1-(1-p_{\text{succ}})^{k}$.

\begin{figure}[ht]
    \centering
    \includegraphics[width=0.65\linewidth]{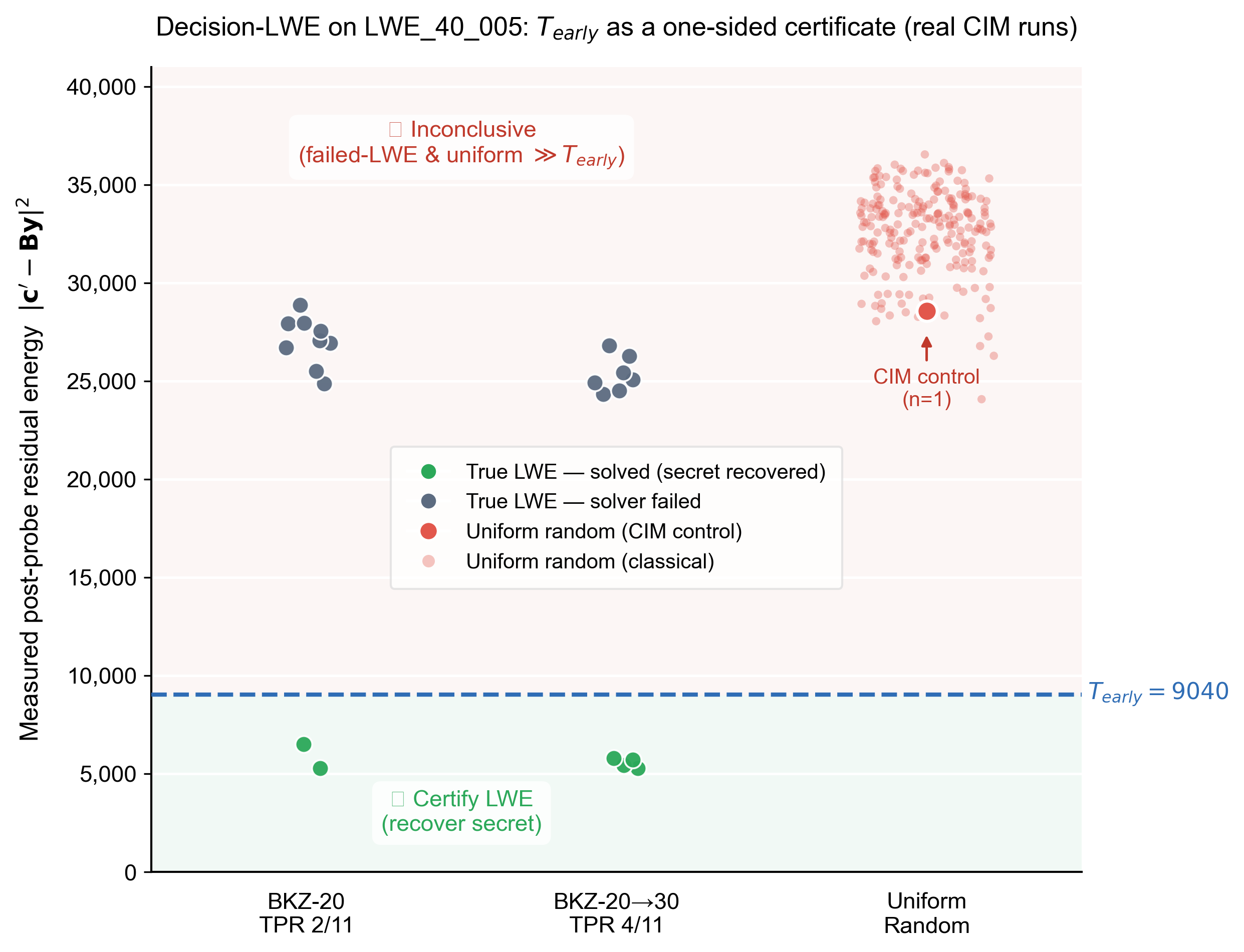}
    \caption{Decision-LWE on \texttt{LWE\_40\_005}: measured post-corrector residual energies
($T_{\text{early}}=9040$). Green = true-LWE seeds the pipeline \emph{solves} (residual
crosses below $T_{\text{early}}$, secret recovered), gray = true-LWE seeds the
solver \emph{fails}, red = uniform-random control, light-red histogram = $n{=}220$
uniform targets evaluated by classical minimization on the same reduced bases.
Progressive BKZ-$20\!\to\!30$ raises the solved fraction from $2/11$ to $4/11$. The
threshold separates \emph{solved} from \emph{unsolved}: only solved seeds cross it, while
failed-LWE and uniform targets are both $\gg T_{\text{early}}$.
Energy magnitude above the threshold is not
discriminative, the operative signal is the binary crossing event.}
    \label{fig:dlwe}
\end{figure}

\section{Computational Complexity and Scalability}
\label{sec:scalability}

Table~\ref{tab:prior_compare} positions our framework against representative families of quantum-assisted lattice cryptanalysis. The two metrics which scale far better than competing NISQ frameworks are \textbf{qubit economy} and \textbf{hardware call count}. HAWI demonstrates DLWE on $n=2$ (a $5$-qubit instance) with a qubit budget scaling as $\le m(m{+}1)=\mathcal{O}(m^2)$ and on the order of $10^4$ shots per expectation value, while Lv \textit{et al.}~\cite{Lv2022} estimate (by resource counting rather than on hardware) about $1{,}126$ logical qubits for an $n=40$ uSVP attack, rising to $13{,}768$ at KYBER-512 scale ($n=256$). Our framework attacks $n=40$ on real hardware with fewer than $20$ logical qubits and a single hardware call with a $36\%$ success probability under progressive BKZ-20$\to$30.

\begin{table}[ht]
\centering
\caption{Comparison of quantum-assisted attacks on LWE/SVP along qubit count and hardware-call count. These two axes are where the present co-design differs most from prior NISQ frameworks. Notes: for Lv~\cite{Lv2022}, the qubit figure is a resource \emph{estimate} at $n=40$ (its actual hardware demonstration is a $3$-qubit toy instance) and its VQE results are simulation-based; for Qayyum~\cite{Qayyum2023, Qayyum2025}, the qubit count is a theoretical upper bound (Prop.~1/2) and the experiments are classical MIP solves rather than runs on a quantum annealer; for Joseph~\cite{Joseph2020}, ``dim'' refers to the quantum SVP simulation.}
\label{tab:prior_compare}
\smallskip
\setlength{\tabcolsep}{3pt}
\renewcommand{\arraystretch}{1.2}
\footnotesize
\resizebox{\textwidth}{!}{%
\begin{tabular}{lllllll}
\toprule
\textbf{Method} & \textbf{Target} & \textbf{Demo.~$n$} & \textbf{Qubits} & \textbf{Penalty} & \textbf{Opt. Loop Scaling} & \textbf{HW calls} \\
\midrule
HAWI~\cite{Zheng2025}                & DLWE              & $n=2$ (HW)            & $\le m(m{+}1)$                & None ($m$ sub-models)      & Gradient descent ($\sim 10$ iters) & $\sim\!10^4$ shots (QAOA) \\
Qayyum~\cite{Qayyum2023, Qayyum2025}  & SLWE \& SIS              & $n\le6$ (MIP)            & $\mathcal{O}(mnq)$ (slack-heavy)                 & Explicit (slack vars)    & QUBO/annealing (proposed) & est.\ only \\
Joseph~\cite{Joseph2020}             & SVP           & dim $\le 4$ (sim.)       & $\mathcal{O}(N \log N)$              & None (1st excited state)      & Sub-adiabatic sweep & repeated \\
Dable-Heath~\cite{DableHeath2023} & SVP & dim $\le 7$ (sim.) & $\mathcal{O}(n\log n)$--$\mathcal{O}(n^{5/2})$ (enc.-dep.) & Restricted qubit ($n$ sub-models; $1$ cyclic) & CIM-CFC sampling & $\sim\!10^3$ runs/inst.\ (sim.) \\
Lv~\cite{Lv2022}                     & SLWE$\to$uSVP                  & $n\le5$ (sim.)           & $\sim$1126 (est., $n{=}40$)     & Explicit  & Variational loop & sim.\ only \\
\midrule
C++ \texttt{fplll} Enum. & CVP                 & $n=40$ (Exact)           & N/A    & None  & $\mathcal{O}(2^d) \cdot \mathcal{O}(d_{\text{froz}}^2) + \mathcal{O}(T_{\text{BKZ}})$ & N/A \\
\textbf{This work}          & \textbf{SLWE \& DLWE}  & $\boldsymbol{n=40}$ (HW) & \textbf{$\le$19} (Logical)     & \textbf{None}           & $ \mathcal{O}(T_{\text{BKZ}})$ & $\boldsymbol{1}$ \textbf{batched submission} \\
\bottomrule
\end{tabular}}
\end{table}

\subsection{Asymptotic Complexity Analysis}
Decomposing the end-to-end pipeline, the total time complexity is
\begin{equation}
\label{eq:total_time}
T_{\text{total}} = \underbrace{T_{\text{BKZ}}}_{\text{global reduction}} 
                 + \underbrace{\mathcal{O}(d_{\text{exp}}^2)}_{\text{CR-BNP}} 
                 + \underbrace{\mathcal{O}(d_{\text{exp}})}_{\text{encoder}} 
                 + \underbrace{\mathcal{O}(1)}_{\text{CIM call}} 
                 + \underbrace{\mathcal{O}(m'^4)}_{\text{bounded corrector } (R\le 2)}.
\end{equation}
The BKZ preprocessing $T_{\text{BKZ}}$ scales as 
$\exp(\mathcal{O}(\beta))$ where $\beta$ is the BKZ block size. The exploration-subspace operations (CR-BNP back-substitution and the adaptive encoder) are at most quadratic in $d_{\text{exp}}$ ($<15$ in our instances). The CIM call is a single evaluation exhibiting constant-time scaling. The bounded corrector iterates over $\binom{m'}{2}$ pairs, each evaluation requires an $\mathcal{O}(m'^2)$ matrix-vector energy reassembly plus an $\mathcal{O}(d_{\text{froz}}\cdot m')$ Babai's nearest plane at the front-end. 

For our $n=40$ configuration ($m'=88$, $d_{\text{exp}}=12$), the time of the bounded corrector is negligible against $T_{\text{BKZ}}$, so the end-to-end wall-clock is dominated by $T_{\text{BKZ}}$, not by the hardware-side operation. Empirically, the single batched CIM submission on the encoded $n=40$ instance takes $0.4$--$9.3$~ms (progressive BKZ-$20\!\to\!30$) across the four successful seeds of Section~\ref{subsec:n40_results}, with logical bit-widths $N_\text{log}\in \{10,13,14,16\}$ produced by the adaptive encoder. The batched $t_{\text{CIM}}$ stays two to four orders of magnitude below the classical reduction time $T_{\text{BKZ}}$ (of order $10^3$\,ms for the progressive schedule). We also timed the optimized C++ \texttt{fplll} solver (\texttt{CVP.closest\_vector}, Schnorr--Euchner enumeration) on the \emph{full} encoded $d_{\text{exp}}=12$ exploration sub-lattice (frozen subspace fixed by Babai): it averages $0.86$\,ms, i.e.\ the CIM is not faster than classical enumeration at the 12-dimension scale. The expectation is that the constant running time of the CIM manifests a scaling advantage as $d_{\text{exp}}$ grows past the point where exhaustive enumeration becomes infeasible.

\subsection{Scalability to Cryptographic Dimensions}
\label{subsec:scalability}
Scaling beyond $n=40$ is an open direction
that may require coordinated adjustments: extending the subspace of the exploration area; widening the search radius (e.g., extending the ternary domain $\pm 1$ to
$\pm 2, \pm 3, \dots$ via wider mixed-radix mappings); replacing the
classical Babai nearest-plane in the frozen subspace with a Lindner--Peikert
\emph{Nearest Planes} enumeration~\cite{LindnerPeikert2011}; and enlarging the BKZ block size to improve basis orthogonality. We leave these for future work.

We also emphasize that this hybrid framework \textit{does not alter the
fundamental asymptotic hardness of LWE}. As Eq.~\eqref{eq:total_time} shows, the
global complexity is dominated by the 
classical BKZ preprocessing, which scales as $\exp(\mathcal{O}(\beta))$. 

\section{Conclusion}
\label{sec:conclusion}
We presented CIM-BDD, a hybrid BDD solver for the
Search and Decision variants of LWE. Its core is a strictly \emph{penalty-free}
QUBO mapping (Section~\ref{sec:modeling}): an algebraic elimination of the secret absorbs the modulus into a primal $q$-ary lattice and recasts LWE as a Closest Vector Problem (especially BDD), then uses the Babai-recentered least-squares residual
\emph{directly} as the QUBO energy. The cryptographic noise thus becomes the
objective rather than a penalized constraint. To fit the NISQ hardware (Section~\ref{sec:hardware}), we use a
CR-BNP-driven adaptive $0/1/2$-bit encoder, block-diagonal partitioning over the
most ambiguous coordinates to make full use of the hardware capacity, a $R\le 2$ corrector as a
deterministic correction layer, and the early-stopping threshold
$T_{\text{early}}$ acts as a one-sided certificate, to recover the secret from the approximate relaxation. 

On the TU Darmstadt LWE Challenge at $n=40$ ($q=1601$, $\sigma=8.005$) with $m'=88$, single-pass Babai fails, whereas the CIM-plus-corrector pipeline recovers the secret on $2$ of $11$ row-seeds ($3$ hits) under baseline BKZ-$20$ and on $4$ of
$11$ ($21$ hits) under progressive BKZ-$20\!\to\!30$, each in a single millisecond-scale
submission.  
Since each seed attacks the same secret and carries a one-sided certificate, the attack stops at the first crossing,  
so the expected number of reductions for progressive BKZ-20$\to$30 is \( 1 / 0.36 \approx 3 \). The  
cumulative recovery probability \( 1 - (1 - 0.36)^k \) reaches 83\% by \( k = 4 \) and exceeds 98\%  
by \( k = 10 \) with independent seeds reducible in parallel.
The whole modeling takes only $N_{\text{log}}\le 19$ logical
bits. Beyond raw recovery rates, the CIM contributes a robust recovery mechanism:
each submission returns an ensemble of low-energy configurations, and the
$R\le 2$ corrector succeeds whenever the true point lies in the neighborhood of any of them---a weaker condition than requiring a single exact
ground-state bit-string. Its batched submission time stays constant in the
millisecond regime (here $0.24$--$17.4$\,ms) irrespective of logical bit-width,
while the end-to-end cost remains dominated by classical BKZ. 

In conclusion, this work provides a methodological template for evaluating post-quantum security primitives on
physical Ising architectures, with the role of the CIM understood as an
ensemble sampler of low-energy states whose practical value at cryptographic
scale remains to be demonstrated. Our framework establishes a scalable foundation for tackling higher-dimensional cryptographic challenges in the future. 


\clearpage

\appendix
\section{Algebraic Isomorphism and Explicit Lattice Basis Extraction}
\label{sec:appendix_algebraic_basis}

In standard reductions from the Learning With Errors (LWE) problem to Bounded-Distance Decoding (BDD), the primal $q$-ary lattice $\Lambda_q(\mathbf{A})$ is often defined using a generating matrix $[\mathbf{A} \mid q\mathbf{I}]$. However, this generating matrix is linearly dependent, so we further construct a full-rank, lower-triangular integer basis for the lattice.

Given the primal LWE matrix $\mathbf{A} \in \mathbb{Z}_q^{m \times n}$ and the target vector $\mathbf{c} \in \mathbb{Z}_q^m$, we randomly sample an $m'$-dimensional subset of equations ($n < m' \le m$). Let this sampled sub-problem be represented by the matrix $\mathbf{A}' \in \mathbb{Z}_q^{m' \times n}$ and the corresponding sub-target vector be $\mathbf{c}' \in \mathbb{Z}_q^{m'}$.

\paragraph{Algebraic Elimination of the Secret.}
To construct the basis for $\Lambda_q(\mathbf{A}')$, we partition the sub-matrix $\mathbf{A}'$ and any corresponding valid lattice vector $\mathbf{v} \in \Lambda_q(\mathbf{A}') \subset \mathbb{Z}^{m'}$ into a top block of dimension $n$ and a bottom block of dimension $m'-n$:
\begin{equation}
    \mathbf{A}' = \begin{bmatrix} \mathbf{A}_{\text{top}} \\ \mathbf{A}_{\text{bot}} \end{bmatrix}, \quad \mathbf{v} = \begin{bmatrix} \mathbf{v}_{\text{top}} \\ \mathbf{v}_{\text{bot}} \end{bmatrix}.
\end{equation}
We enforce that the leading $n \times n$ submatrix $\mathbf{A}_{\text{top}}$ is invertible modulo $q$, resampling the $m'$ rows if this condition fails. By the definition of the $q$-ary lattice, any valid lattice vector $\mathbf{v}$ evaluated modulo $q$ must satisfy the congruence relation $\mathbf{v} \equiv \mathbf{A}'\mathbf{s} \pmod q$ for some secret $\mathbf{s} \in \mathbb{Z}_q^n$. Decomposing this congruence yields:
\begin{align}
    \mathbf{v}_{\text{top}} &\equiv \mathbf{A}_{\text{top}} \mathbf{s} \pmod q, \label{eq:v_top_app} \\
    \mathbf{v}_{\text{bot}} &\equiv \mathbf{A}_{\text{bot}} \mathbf{s} \pmod q. \label{eq:v_bot_app}
\end{align}

Because $\mathbf{A}_{\text{top}}$ is explicitly chosen to be invertible over $\mathbb{Z}_q$, we can uniquely express the hidden cryptographic variable $\mathbf{s}$ from Eq.~\eqref{eq:v_top_app}:
\begin{equation} \label{eq:s_elim_app}
    \mathbf{s} \equiv \mathbf{A}_{\text{top}}^{-1} \mathbf{v}_{\text{top}} \pmod q.
\end{equation}
Substituting Eq.~\eqref{eq:s_elim_app} into Eq.~\eqref{eq:v_bot_app} eliminates the secret $\mathbf{s}$ entirely, yielding a deterministic algebraic constraint between the top and bottom coordinates of any valid lattice point:
\begin{equation}
    \mathbf{v}_{\text{bot}} \equiv \mathbf{A}_{\text{bot}} (\mathbf{A}_{\text{top}}^{-1} \mathbf{v}_{\text{top}}) \equiv \mathbf{T}_{\text{mat}}\mathbf{v}_{\text{top}} \pmod q,
\end{equation}
where $\mathbf{T}_{\text{mat}} = \mathbf{A}_{\text{bot}}\mathbf{A}_{\text{top}}^{-1} \pmod q \in \mathbb{Z}_q^{(m'-n) \times n}$ acts as the transition matrix.

\paragraph{Basis Construction.}
To lift this modular congruence into the discrete integer domain $\mathbb{Z}^{m'}$, we introduce an arbitrary integer vector $\mathbf{k} \in \mathbb{Z}^{m'-n}$ to absorb the modulo $q$. This converts the congruence into an integer equality:
\begin{equation} \label{eq:integer_lift_app}
    \mathbf{v}_{\text{bot}} = \mathbf{T}_{\text{mat}}\mathbf{v}_{\text{top}} + q\mathbf{k}.
\end{equation}

This establishes a necessary and sufficient condition: a point $\mathbf{v} \in \mathbb{Z}^{m'}$ mathematically belongs to the lattice $\Lambda_q(\mathbf{A}')$ if and only if it can be parameterized by arbitrary free integer variables $\mathbf{v}_{\text{top}} \in \mathbb{Z}^n$ and $\mathbf{k} \in \mathbb{Z}^{m'-n}$ as follows:
\begin{equation} \label{eq:basis_construction_app}
    \mathbf{v} = \begin{bmatrix} \mathbf{v}_{\text{top}} \\ \mathbf{v}_{\text{bot}} \end{bmatrix} 
    = \begin{bmatrix} \mathbf{v}_{\text{top}} \\ \mathbf{T}_{\text{mat}}\mathbf{v}_{\text{top}} + q\mathbf{k} \end{bmatrix}
    = \underbrace{\begin{bmatrix} \mathbf{I}_{n \times n} & \mathbf{0}_{n \times (m'-n)} \\ \mathbf{T}_{\text{mat}} & q \mathbf{I}_{(m'-n) \times (m'-n)} \end{bmatrix}}_{\mathbf{B}_{\text{init}}} 
    \begin{bmatrix} \mathbf{v}_{\text{top}} \\ \mathbf{k} \end{bmatrix}.
\end{equation}

Extracting the coefficients of the free integer vector $[\mathbf{v}_{\text{top}}^\top, \mathbf{k}^\top]^\top$ immediately provides the explicit block matrix representation of the initial lattice basis $\mathbf{B}_{\text{init}} \in \mathbb{Z}^{m' \times m'}$.
$\mathbf{B}_{\text{init}}$ is a full-rank integer basis ready to be globally reduced by BKZ/LLL algorithms into $\mathbf{B}_0$.

\section{Detailed Algebraic Formulations and Execution Pipeline}
\label{sec:appendix_pipeline}

This appendix provides mathematical formulations in the CIM-BDD framework, progressing sequentially from global preprocessing to cryptographic recovery. Algorithm~\ref{alg:qcbdd} synthesizes the step-by-step mathematical flow.

\subsection{Pre-processing and Hardware Embedding(Algorithm Phases 1--4)}

\paragraph{Step 1: Algebraic Isomorphism and Global Reduction.}
Given the primal LWE matrix \( \mathbf{A} \in \mathbb{Z}_q^{m \times n} \) and the target vector \( \mathbf{c} \in \mathbb{Z}_q^m \), we perform randomized trials to uniformly sample an \( m' \)-dimensional subset of rows without replacement. We enforce that its leading \( n \times n \) submatrix \( \mathbf{A}_{\text{top}} \in \mathbb{Z}_q^{n \times n} \) is invertible modulo \( q \), resampling the subset if this condition fails. The transition matrix is computed algebraically as $\mathbf{T}_{\text{mat}} \equiv \mathbf{A}_{\text{bot}}\mathbf{A}_{\text{top}}^{-1} \pmod q \in \mathbb{Z}_q^{(m'-n) \times n}$. 

This isomorphism systematically constructs the primal $q$-ary lattice $\Lambda_q(\mathbf{A})$. Any valid lattice vector $\mathbf{v} \in \Lambda_q(\mathbf{A})$ satisfies the modulo relation $\mathbf{v}_{\text{bot}} \equiv \mathbf{T}_{\text{mat}}\mathbf{v}_{\text{top}} \pmod q$. By expressing this modular congruence as an integer linear combination $\mathbf{v}_{\text{bot}} = \mathbf{T}_{\text{mat}}\mathbf{v}_{\text{top}} + q \cdot \mathbf{k}$ for some arbitrary integer vector $\mathbf{k} \in \mathbb{Z}^{m'-n}$, we explicitly formulate the block lower-triangular CVP basis $\mathbf{B}_{\text{init}} \in \mathbb{Z}^{m' \times m'}$ as the block matrix:
\begin{equation}
    \mathbf{B}_{\text{init}} = \begin{bmatrix} \mathbf{I}_{n \times n} & \mathbf{0}_{n \times (m'-n)} \\ \mathbf{T}_{\text{mat}} & q \cdot \mathbf{I}_{(m'-n) \times (m'-n)} \end{bmatrix} \in \mathbb{Z}^{m' \times m'}.
\end{equation}
Applying global lattice reduction (e.g., BKZ) performs integer unimodular transformations on $\mathbf{B}_{\text{init}}$, yielding the highly orthogonalized integer basis $\mathbf{B}_0 \in \mathbb{Z}^{m' \times m'}$. The target vector maps unchanged as a target point $\mathbf{c}'$.

\paragraph{Step 2: Gram-Schmidt Projection and Orthogonal Slicing.}
Global lattice reduction implicitly maintains the Gram-Schmidt Orthogonalization (GSO), decomposing the integer basis into continuous real matrices: $\mathbf{B}_0 = \mathbf{B}^* \mathbf{U}_\mu$, where $\mathbf{B}^* \in \mathbb{R}^{m' \times m'}$ is the orthogonal basis matrix, and $\mathbf{U}_\mu \in \mathbb{R}^{m' \times m'}$ is the upper unit-triangular matrix ($\mathbf{U}_\mu[i,i]=1$).

To construct an analytically uncoupled orthogonal coordinate system, we project the discrete target $\mathbf{c}'$ onto the continuous orthogonal basis $\mathbf{B}^*$. Let $\mathbf{c}' = \sum_{j=1}^{m'} \alpha_j \mathbf{b}_j^*$, where $\boldsymbol{\alpha} \in \mathbb{R}^{m'}$ is the continuous projection coefficient vector, $\alpha_i = \frac{\langle \mathbf{c}', \mathbf{b}_i^* \rangle}{\|\mathbf{b}_i^*\|^2} \in \mathbb{R} \quad \forall i \in \{1, \dots, m'\}$.

To geometrically decouple the exploration zone (the trailing $d_{\text{exp}}$ dimensions, denoted as the orthogonal complement subspace $\mathcal{F}^\perp$), we extract the residual continuous target projection $\boldsymbol{\alpha}_{\text{exp}} \in \mathbb{R}^{d_{\text{exp}}}$ and the bottom-right submatrix $\mathbf{U}_{\mu,\text{exp}} \in \mathbb{R}^{d_{\text{exp}} \times d_{\text{exp}}}$. Because $\mathbf{U}_{\mu}$ is structurally upper-triangular, this trailing continuous subspace is decoupled from low-index variables.

\paragraph{Step 3: Continuous Relaxed Babai's Nearest Plane (CR-BNP).}
We apply a continuous relaxation by suspending the discrete rounding operations to locate the continuous intersection point within \( \mathcal{F}^{\perp} \).
This translates to solving the continuous upper-triangular linear system $\mathbf{U}_{\mu,\text{exp}} \mathbf{y}_c = \boldsymbol{\alpha}_{\text{exp}}$ to obtain the relaxed fractional coordinates $\mathbf{y}_c \in \mathbb{R}^{d_{\text{exp}}}$. By expanding the $i$-th row of this linear system, we directly obtain the $\mathcal{O}(d_{\text{exp}}^2)$ Gaussian back-substitution formula:
\begin{equation}
    (\mathbf{y}_c)_i + \sum_{j=i+1}^{d_{\text{exp}}} U_{\mu,\text{exp},i,j} \cdot (\mathbf{y}_c)_j = \alpha_{\text{exp},i} \quad \implies \quad (\mathbf{y}_c)_i = \alpha_{\text{exp},i} - \sum_{j=i+1}^{d_{\text{exp}}} U_{\mu,\text{exp},i,j} \cdot (\mathbf{y}_c)_j.
\end{equation}
We then establish a discrete integer center $\mathbf{y}_{\text{center}} = \lfloor \mathbf{y}_c \rceil \in \mathbb{Z}^{d_{\text{exp}}}$. The alignment error is isolated as the continuous fractional ambiguity $\boldsymbol{\delta} = \mathbf{y}_c - \mathbf{y}_{\text{center}} \in [-0.5, 0.5]^{d_{\text{exp}}}$, which dictates the subsequent physical hardware mapping.

\paragraph{Step 4: Adaptive Encoding and Subspace QUBO Formulation.}
Transfer the continuous subspace ($\mathbb{R}$) back into the discrete domain ($\mathbb{Z}$) for hardware execution. Let $\mathbf{z} \in \{0,1\}^{N_{\text{log}}}$ denote the logical binary variables evaluated by CIM. Based on the magnitude and sign of the continuous ambiguity $\boldsymbol{\delta}$, an adaptive encoder constructs a non-negative integer step-size encoding matrix $\mathbf{T}_{\text{enc}}$ and a deterministic base shift vector $\boldsymbol{\mathrm{offset}}_{\text{enc}}$, governed by empirical thresholds $\tau_{\text{low}}$, $\tau_{\text{high}}$:
\begin{itemize}
    \item \textbf{0-Bit ($|\delta_i| < \tau_{\text{low}}$):} We fix $\text{offset}_{\text{enc},i} = 0$, appending no columns to $\mathbf{T}_{\text{enc}}$.
    \item \textbf{1-Bit ($|\delta_i| \ge \tau_{\text{high}}$):} A single bit explores the most likely adjacent integer. We fix $\text{offset}_{\text{enc},i} = 0$, and append one column containing the value $\operatorname{sgn}(\delta_i)$ to $\mathbf{T}_{\text{enc}}$.
    \item \textbf{2-Bit ($\tau_{\text{low}} \le |\delta_i| < \tau_{\text{high}}$):} To symmetrically explore the local integer neighborhood $\{-1, 0, 1\}$, we explicitly set a negative shift $\text{offset}_{\text{enc},i} = -1$, and append two independent columns to $\mathbf{T}_{\text{enc}}$ with the value $+1$.
\end{itemize}

The absolute search floor is fixed as $\mathbf{y}_{\text{base}} = \mathbf{y}_{\text{center}} + \boldsymbol{\mathrm{offset}}_{\text{enc}} \in \mathbb{Z}^{d_{\text{exp}}}$. Consequently, the candidate displacement vector generated by the hardware is constrained to the integer lattice domain:
$ \tilde{\mathbf{y}}_{\text{exp}} = \mathbf{y}_{\text{base}} + \mathbf{T}_{\text{enc}}\mathbf{z} \in \mathbb{Z}^{d_{\text{exp}}}
$

To construct a QUBO landscape, we evaluate the true squared Euclidean distance in the orthogonal continuous space. The coordinate continuous error relative to the basis $\mathbf{B}^*$ is $\boldsymbol{\Delta} = \boldsymbol{\alpha}_{\text{exp}} - \mathbf{U}_{\mu,\text{exp}}\tilde{\mathbf{y}}_{\text{exp}} \in \mathbb{R}^{d_{\text{exp}}}$. 

Because orthogonal basis vectors have varying non-unit lengths, computing the true Euclidean norm requires weighting each squared coordinate error $\Delta_i^2$ by its squared basis length $\|\mathbf{b}_i^*\|^2$. Let $\mathbf{D}_{\text{exp}} = \mathrm{diag}(\|\mathbf{b}_{m'-d_{\text{exp}}+1}^*\|^2, \dots, \|\mathbf{b}_{m'}^*\|^2) \in \mathbb{R}^{d_{\text{exp}} \times d_{\text{exp}}}$ be the continuous diagonal weighting matrix. The penalty-free continuous energy $E_{\text{proj}}(\mathbf{z}) \in \mathbb{R}^{\ge 0}$ optimized by QUBO evaluates to:
\begin{equation}
\begin{aligned}
    E_{\text{proj}}(\mathbf{z}) &= \sum_{i=1}^{d_{\text{exp}}} \Delta_i^2 \|\mathbf{b}_{m'-d_{\text{exp}}+i}^*\|^2 = \boldsymbol{\Delta}^\top \mathbf{D}_{\text{exp}} \boldsymbol{\Delta} \\
    &= \left(\boldsymbol{\alpha}_{\text{exp}} - \mathbf{U}_{\mu,\text{exp}}\tilde{\mathbf{y}}_{\text{exp}}\right)^\top \mathbf{D}_{\text{exp}} \left(\boldsymbol{\alpha}_{\text{exp}} - \mathbf{U}_{\mu,\text{exp}}\tilde{\mathbf{y}}_{\text{exp}}\right) \\
    &= \left(\boldsymbol{\alpha}_{\text{exp}} - \mathbf{U}_{\mu,\text{exp}}(\mathbf{y}_{\text{base}} + \mathbf{T}_{\text{enc}}\mathbf{z})\right)^\top \mathbf{D}_{\text{exp}} \left(\boldsymbol{\alpha}_{\text{exp}} - \mathbf{U}_{\mu,\text{exp}}(\mathbf{y}_{\text{base}} + \mathbf{T}_{\text{enc}}\mathbf{z})\right).
\end{aligned}
\end{equation}

Before mapping to the physical hardware, the derived continuous external fields and couplings are linearly scaled independently for each sub-QUBO into the 8-bit signed integer range supported by the CIM.

\subsection{Post-Processing and Verification (Phase 5)}

\paragraph{Step 5: Frozen-Subspace Babai's Nearest Plane.}
Upon hardware convergence, the CIM returns a discrete bitstring $\mathbf{z}_{\text{hw}} \in \{0,1\}^N$. The true lattice displacement in the exploration zone is reconstructed as an integer vector $\mathbf{y}_{\text{exp}} = \mathbf{y}_{\text{base}} + \mathbf{T}_{\text{enc}}\mathbf{z}_{\text{hw}} \in \mathbb{Z}^{d_{\text{exp}}}$. Substituting this configuration into the global lattice, a classical Babai's Nearest Plane sequentially computes the compensatory frozen integer coordinates $\mathbf{y}_{\text{froz}} \in \mathbb{Z}^{m'-d_{\text{exp}}}$ over $\mathbf{B}_{\text{froz}}$, subject to the dynamically updated integer target $\mathbf{c}_{\text{froz}} = \mathbf{c}' - \mathbf{B}_{0,\text{exp}}\mathbf{y}_{\text{exp}} \in \mathbb{Z}^{m'}$.

\paragraph{Step 6: Bounded Corrector and Secret Recovery.}
Concatenating the segments yields the global discrete candidate vector $\mathbf{y}_{\text{cand}} = [\mathbf{y}_{\text{froz}}^\top, \mathbf{y}_{\text{exp}}^\top]^\top \in \mathbb{Z}^{m'}$. The global discrete residual energy is $E = \|\mathbf{c}' - \mathbf{B}_0\mathbf{y}_{\text{cand}}\|^2$. 

If $E > T_{\text{early}}$, the deterministic bounded corrector ($R \le 2$) operates on the full candidate vector to search adjacent integer assignments.

Upon successfully capturing an energy state $E \le T_{\text{early}}$, the noise-free lattice vector $\mathbf{v} = \mathbf{B}_0 \mathbf{y}_{\text{cand}} \in \mathbb{Z}^{m'}$ is verified. Finally, the true LWE cryptographic secret $\mathbf{s} \in \mathbb{Z}_q^n$ is algebraically recovered by solving $\mathbf{s} \equiv \mathbf{A}_{\text{top}}^{-1}\mathbf{v}_{\text{top}} \pmod q$.

\end{document}